\documentclass[11 pt, onecolumn, journal, draftcls]{IEEEtran}
\usepackage{balance}
\usepackage{cite}      % Written by Donald Arseneau
\ifx\pdfoutput\undefined
\usepackage{graphicx}
\else
\usepackage[pdftex]{graphicx}
\fi

\usepackage{psfrag}    % Written by Craig Barratt, Michael C. Grant,
\usepackage{subfigure} % Written by Steven Douglas Cochran

\usepackage{amsmath}   % From the American Mathematical Society
\usepackage{amssymb}
\usepackage{latexsym}
\usepackage{amsfonts}
\newcommand{\ignore}[1]{}

\newtheorem{theorem}{Theorem}%[section]
\newtheorem{lemma}[theorem]{Lemma}

\newtheorem{definition}[theorem]{Definition}

\newcommand{\qed}{\hfill $\Box$}
\newenvironment{proof}{\par\noindent{\noindent \bf Proof:}}{\qed \par}
\newcommand{\xor }[0]{\oplus }

\newcommand{\zo }[0]{\{0,1\} }

\renewcommand{\P}[0]{{\cal P}}
\newcommand{\E}[0]{{\cal E}}

\newcommand{\Y}[0]{{\cal{Y}}}

\newcommand{\A}{{\cal A}}
\newcommand{\ls}[1]
   {\dimen0=\fontdimen6\the\font \lineskip=#1\dimen0
\advance\lineskip.5\fontdimen5\the\font \advance\lineskip-\dimen0
\lineskiplimit=.9\lineskip \baselineskip=\lineskip
\advance\baselineskip\dimen0 \normallineskip\lineskip
\normallineskiplimit\lineskiplimit \normalbaselineskip\baselineskip
\ignorespaces }

\newcommand{\remove}[1]{}

\newcounter{alpha:count}
{\begin{list}{{\alph{alpha:count}.}}%
{\usecounter{alpha:count}\setlength{\listparindent}{0pt}}}%
{\end{list}}

\newcounter{roman:count}
{\begin{list}{{\roman{roman:count}.}}%
{\usecounter{roman:count}\setlength{\listparindent}{0pt}}}%
{\end{list}}

\newcommand{\eps}{\varepsilon}

\newcommand{\set}[1]{\{#1\}}

\newcommand{\K}[0]{{\cal K}}

\def\inr{\in_R}

\def\ENC{Enc}
\def\DEC{Dec}

\def\P{\cal P}
\def\Q{\cal Q}
\def\X{\cal X}
\def\REAL{\mbox{\sc real}}
\def\RP{\mbox{\sc rp}}
\def\Preal{P_{REAL}}
\def\Prp{P_{RP}}
\def\Dist{Dist}
\def\Cinv{Y_c}

\date{September 22, 2009}

\makeatother

\begin{document} 
\title{On Compression of Data Encrypted with Block Ciphers}

\author{Demijan~Klinc, Carmit~Hazay, Ashish~Jagmohan, Hugo~Krawczyk, and Tal~Rabin
\thanks{D. Klinc is with the Georgia Institute of Technology, Atlanta, GA. Email: demi@ece.gatech.edu}%
\thanks{C. Hazay is with Bar-Ilan University, Ramat-Gan, Israel. Email: harelc@cs.biu.ac.il}%
\thanks{A. Jagmohan, H. Krawczyk, and T. Rabin are with IBM T.J. Watson Research Labs, Yorktown Heights and Hawthorne, NY. %
Email:~\{ashishja, talr\}@us.ibm.com, hugo@ee.technion.ac.il}}
%\markboth{IEEE Transactions on Information Theory}{}
\maketitle

\begin{abstract}
This paper investigates compression of data encrypted with block ciphers, such as the Advanced Encryption Standard (AES).  It is shown that such data can be feasibly compressed without knowledge of the secret key. Block ciphers operating in various chaining modes are considered and it is shown how compression can be achieved without compromising security of the encryption scheme. Further, it is shown that there exists a fundamental limitation to the practical compressibility of block ciphers when no chaining is used between blocks. Some performance results for practical code constructions used to compress binary sources are presented.
\end{abstract}


\section{Introduction}

We consider the problem of compressing encrypted data. Traditionally in communication systems, data from 
a source is first compressed and then encrypted before it is transmitted over a channel to the receiver. While in many cases this 
approach is befitting, there exist scenarios where there is a need to reverse the order in which data encryption and compression 
are performed. Consider for instance a network of low-cost sensor nodes that transmit sensitive information over 
the internet to a recipient. The sensor nodes need to encrypt data to hide it from potential eavesdroppers, but they may not be able to perform compression as that would require additional hardware and thus higher implementation cost. On the other hand, the network operator that is responsible for transfer of data to the recipient wants to compress the data to maximize the utilization of its resources. It is important to note that the network operator is not trusted and hence does not have access to the  key used for encryption and decryption. If it had the key, it could simply decrypt data, compress and encrypt again.

We focus on compression of encrypted data where the encryption procedure utilizes block ciphers such as the Advanced Encryption Standard (AES)~\cite{US01AES}  and Data Encryption Standard (DES)~\cite{NBS77Data}. Loosely speaking, block ciphers operate on inputs of fixed length and serve as important building blocks that can be used to construct secure encryption schemes.

%%\hugo{Replace the reference to AES from Wenbo's book to the following (There is one more place in the file where the same change of citation is needed; search Mao)
%%\\
%%               U.S. Department of Commerce/National Institute of
%%                 Standards and Technology, "Advanced Encryption Standard
%%                 (AES)", FIPS PUB 197, November 2001,
%%                \\ {\tt <http://csrc.nist.gov/publications/fips/index.html>.}
%}

For a fixed key a block cipher is a bijection, therefore the entropy of an input is the same as that of the output. It follows that it is theoretically possible to compress the source to the same level as before encryption. However, in practice, encrypted data appears to be random and the conventional compression techniques do not yield desirable results. Consequently, it was long believed that encrypted data is practically incompressible. In a surprising 
paper~\cite{Johnson04OnCompressing}, the authors break that paradigm and show that the problem of compressing one-time pad encrypted data translates to the problem of compressing correlated sources, which was solved by Slepian and Wolf in~\cite{Slepian73Noiseless} and for which practical and efficient codes are known. Compression is practically achievable due to a simple symbol-wise correlation between the key (one-time pad) and the encrypted message. However, when such correlation is more complex, as is the case with block ciphers, the approach to Slepian-Wolf coding utilized in~\cite{Johnson04OnCompressing} is not directly applicable.

In this paper, we investigate if data encrypted with block ciphers can be
compressed without access to the key. We show that block ciphers in conjunction
with the most commonly used chaining modes in practice (e.g., \cite{Kent05Security, Dierks08TLS}) 
%KEEP THIS NOTE: TLS' mandatory ciphersuite is TLS_RSA_WITH_AES_128_CBC_SHA according to section 9 of the TLS RFC 5246
are practically compressible for some types of sources. To our knowledge this is the first work to show that non-negligible compression gains can be achieved for cryptographic algorithms like AES or DES when they are used in non-stream modes; in particular, this work offers a solution to the open problem formulated in~\cite[Sec. 3.3]{Johnson04OnCompressing2}.

Moreover, by using standard techniques in the cryptographic
literature we prove that the proposed compression schemes do not compromise
the security of the original encryption scheme.  
%The proof follows standard techniques in the cryptographic literature and uses the fact that both compression and decompression algorithms access the encryption and decryption functionalities as black boxes (rather than utilizing the internals of the block cipher).  
%HUGO: The above sentence relates to a more general proof of security for any
%coding scheme with black box access to ENC/DEC. We are not pursuing this
%general result here so I erased the sentence. Also, note that for CPA security
%we are basically making no use of the DEC functionality. Yet, it is worth
%remembering the following subtlety: Since we are assuming that D can return a
%value other than the real plaintext (we allow for a small decompression error
%\eps), one could build a D that outputs with some small probability (smaller
%than \eps) the decryption key, without contradicting the definition of (C,D). 
%If a scheme uses DEC as a black-box it cannot do that since it does not have
%the key. 
We also show that there exists a fundamental limitation to the compression
capability when the input to the block cipher is applied to a single-block
message without chaining as in ECB mode (see Section \ref{ssec:ecb}). 

The outline of this paper is as follows. Section~\ref{sec:background} defines the problem that we seek to solve and summarizes existing work on the subject. Section~\ref{sec:blockCiphers} focuses on block ciphers and explains how they can be compressed without knowledge of the secret key in various modes of operation. Section~\ref{sec:security} discusses security of the proposed compression scheme, while Section~\ref{sec:publicEncryptionSchemes} touches on the subject of compressing data encrypted with public-key encryption schemes. Section~\ref{sec:simResults} presents some simulation results for binary memoryless sources and finally, Section~\ref{sec:conclusion} concludes the paper.

\section{Preliminaries}
\label{sec:background}

We begin with a standard formal definition of an encryption scheme as stated in~\cite{Katz07Introduction}.
A private-key encryption scheme is a triple of algorithms $(Gen,\ENC,\DEC)$, where
$Gen$ is a probabilistic algorithm that outputs a key $K$ chosen according to some distribution that is determined by the scheme;
the encryption algorithm $\ENC$ takes as input a key $K$ and a plaintext message $X$ and outputs a ciphertext $\ENC_K(X)$;
the decryption algorithm $\DEC$ takes as input a key $K$ and a ciphertext $\ENC_K(X)$ and outputs a plaintext $\DEC_K(\ENC_K(X)) = X$. It is required that for every key $K$
output by $Gen$ and every plaintext $X$, we have $\DEC_K(\ENC_K(X)) = X$.

In private-key encryption schemes of concern to us in this paper the same key is used for encryption and decryption algorithms. Private-key encryption schemes can be divided in two categories: block ciphers and stream ciphers. Stream ciphers encrypt plaintext one symbol at a time, typically by summing it with a key (XOR operation for binary alphabets). In contrast, block ciphers represent a different approach where encryption is accomplished by means of nonlinear mappings on input blocks of fixed length. Common examples of block ciphers are AES and DES.
Typically, block ciphers are not used as a stand-alone encryption procedure. Instead, they
are combined to work on variable-length data using composition
mechanisms known as chaining modes or modes of operation, as specified
in Section~\ref{sec:blockCiphers}.

We proceed with a formulation of the source coding problem with decoder side-information, which
is illustrated in Figure~\ref{fig:swc}. Consider random variables
$X$ (termed the source), and $S$ (termed the side-information),
both over a finite-alphabet and with a joint probability distribution $P_{XS}$. Consider a
sequence of independent $n$ realizations of $(X,S)$ denoted by $\{ X_{i},S_{i}\}_{i=1}^{n}$ .

% \ignorefigure{
\begin{figure}[htb!]
\centering
\includegraphics[width=3.5in]{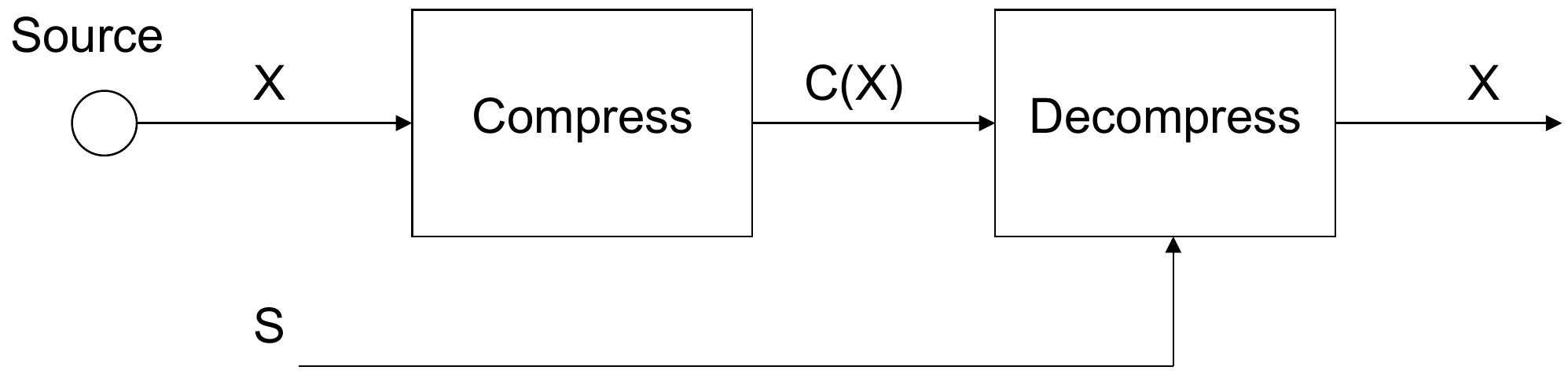}
\caption{Lossless source coding with decoder side-information.}
\label{fig:swc}
\end{figure}
% }%ignorefigure

The problem at hand is that of losslessly encoding $\{ X_{i}\}_{i=1}^{n}$,
with $\{ S_{i}\}_{i=1}^{n}$ known only to the decoder. In \cite{Slepian73Noiseless},
Slepian and Wolf showed that, for sufficiently large block length $n$, this
can be done at rates arbitrarily close to the conditional entropy
$H(X|S).$ Practical Slepian-Wolf coding schemes use constructions
based on good linear error-correcting codes~\cite{Aaron02Compression, Frias01Compression, Liveris02Compression}.

% \ignorefigure{
\begin{figure}[t]
\centering
	\subfigure[]
	{
		\label{fig:enc_comp_a}
		\includegraphics[width=2.8in]{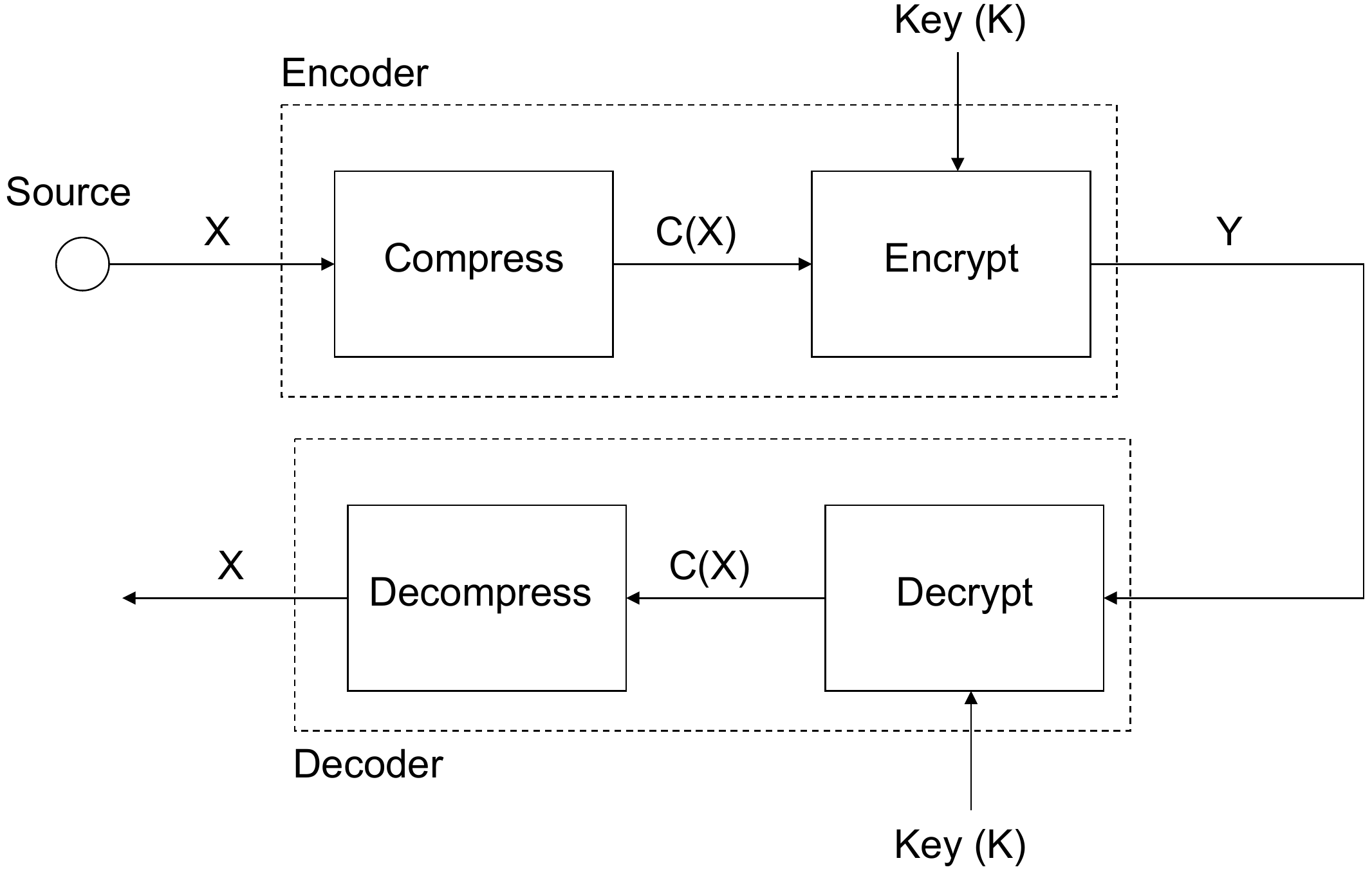}
	}
	\subfigure[]
	{
		\label{fig:enc_comp_b}
		\includegraphics[width=2.8in]{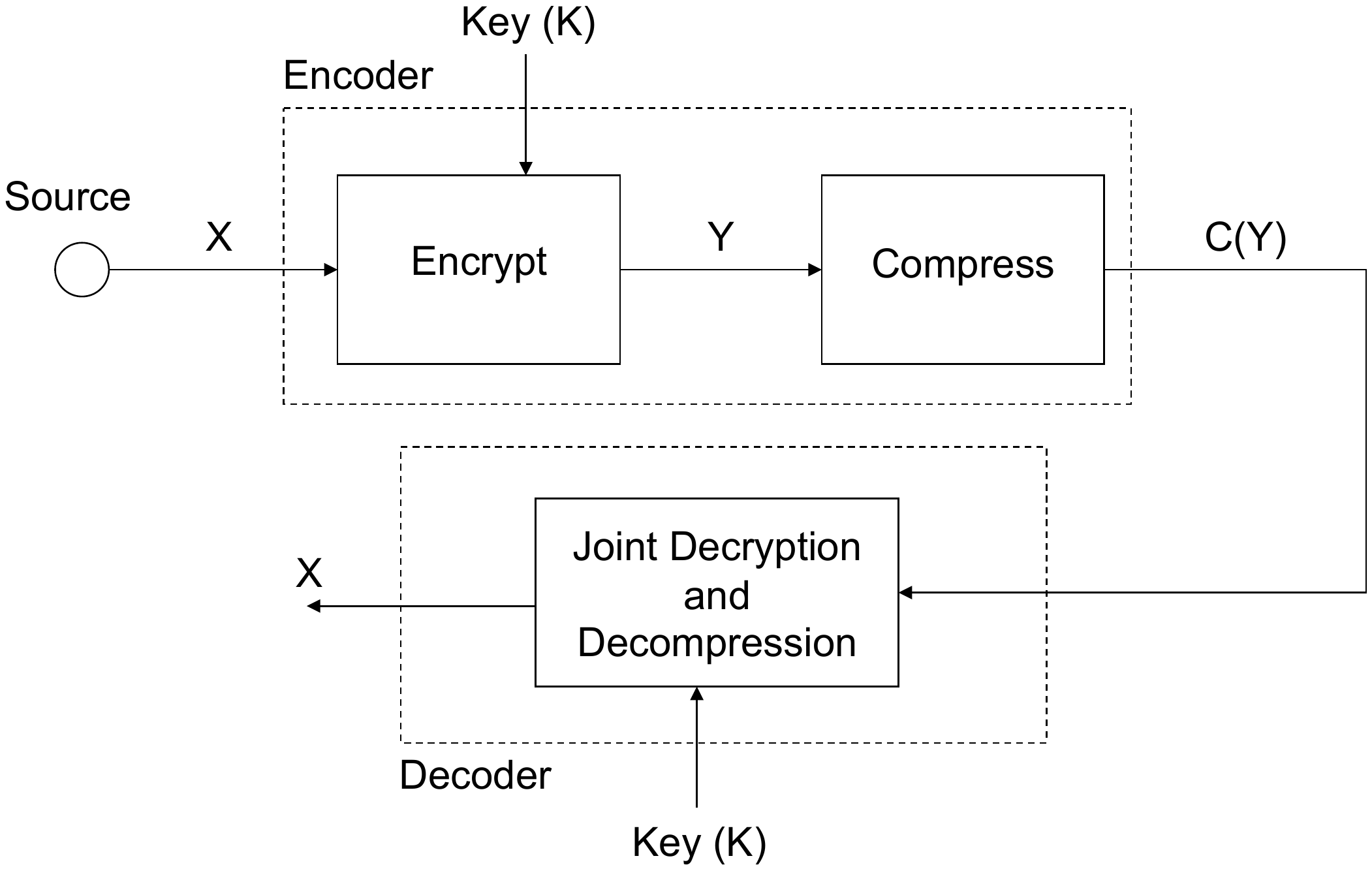}
	}
	\caption{Systems combining compression and encryption. (a) Traditional system 
		with compression done first. (b) System with encryption done before compression.}
  	\label{fig:enc_comp}
\end{figure}
% }%ignorefigure

% \ignorefigure{
%\begin{figure}[htbp]
%
% \begin{minipage}[c]{3in}
%    \centering
%    \includegraphics[width=2.8in]{fig_trad_sys.pdf}  
%    \caption{(a)}
%  \end{minipage}
%  \begin{minipage}[c]{3in}
%    \centering
%    \includegraphics[width=2.8in]{fig_new_sys.pdf}  
%    \caption{(b)}
%  \end{minipage}
%  \caption{Systems combining compression and encryption. (a) Traditional system 
%  with compression done first. (b) System with encryption done before compression.}
%  \label{fig:enc_comp}
%\end{figure}
% }%ignorefigure

In this paper, we are interested in systems that perform both compression
and encryption, wherein the compressor does not have access to the
key. Typically, in such systems, encryption is performed after compression
as depicted in Figure \ref{fig:enc_comp_a}. This is a consequence
of the traditional view which considers ciphertext data hard to compress
without knowledge of the key. In \cite{Johnson04OnCompressing} a system similar
to Figure \ref{fig:enc_comp_b} is considered instead, in which the
order of the encryption and compression operations at the encoder
is reversed (note, though, that only the encryptor has access to the key). 
The authors consider encryption of a plaintext $X$ using a one-time
pad scheme, with a finite-alphabet key (pad) K, to generate
the ciphertext $Y$, i.e. 
\begin{displaymath}
Y_i\triangleq X_i\oplus K_i. %,\,\,\,\forall j\in{\mathbb{{Z}}}.
\end{displaymath}
%\hugo{I guess that either i or j need to be changed; actually, why use a
%subscript at all here? and why over ${\mathbb{{Z}}}$?}
\noindent This is followed by compression, which is agnostic
of $K$, to generate the compressed ciphertext $C(Y)$. 

The key insight underlying the approach in \cite{Johnson04OnCompressing} is
that the problem of compression in this case can be formulated as
a Slepian-Wolf coding problem. In this formulation the ciphertext
$Y$ is cast as a source, and the shared key $K$ is cast as
the decoder-only side-information. The joint distribution of the source
and side-information can be determined from the statistics of the
source. For example, in the binary case with a uniformly distributed $K_i$
and $X_i$ with $\Pr[X_i=1] = p$, we have
\begin{equation}
\Pr(Y_i\neq k|K_i=k)=p.\label{eq:corr}
\end{equation}
%\hugo{why does Y has a subscript here and K does not?}
\noindent The decoder has knowledge of $K$ and of the source statistics.
It uses this knowledge to reconstruct the ciphertext $Y$ from
the compressed message $C(Y)$, and to subsequently decrypt
the plaintext $X$. This formulation is leveraged in \cite{Johnson04OnCompressing}
to show that exactly the same lossless compression rate, $H(X)$,
can be asymptotically achieved in the system shown in Figure \ref{fig:enc_comp_b},
as can be achieved in Figure \ref{fig:enc_comp_a}. Further, this
can be done while maintaining information-theoretic security.

The one-time pad and stream ciphers, while convenient for analysis, are not the only forms of encryption in practice. 
In fact, the prevalent method of encryption uses block ciphers, thus an obviously desirable extension of the technique
in \cite{Johnson04OnCompressing} would be to conventional encryption schemes such
as the popular AES encryption method. Attempting to do so, however,
proves to be problematic. The method in \cite{Johnson04OnCompressing} leverages
the fact that in a one-time pad encryption scheme there exists a simple symbol-wise
correlation between the key $K$ and the ciphertext $Y$,
as seen in (\ref{eq:corr}). Unfortunately, for block ciphers 
such as AES no such correlation structure is known. Moreover, any
change in the plaintext is diffused in the ciphertext, and quantifying
the correlation (or the joint probability distribution) of the key
and the ciphertext is believed to be computationally infeasible.

In the remainder of this paper, we show how this problem can be circumvented
by exploiting the chaining modes popularly used with block ciphers. Based on
this insight, we present an approach for compressing data encrypted
with block ciphers, without knowledge of the key. As in \cite{Johnson04OnCompressing},
the proposed methods are based on the use of Slepian-Wolf coding. 

Before continuing, we formalize the notion of a ``post-encryption compression"
which is used througout this paper with the following definition:

\smallskip
\begin{definition}[Post-Encryption Compression (PEC) Scheme]\label{def-ped}

Let $\E=(Gen,\ENC,\DEC)$ be an encryption scheme, as defined above, with plaintext domain $\X$ and
ciphertext range $\Q$, and let $\P$ be a probability distribution over $\X$.
Let $C$ be a compression function defined over $\Q$, and $D$ be a (possibly
probabilistic) decoding
function with the property that, for any key $K$ generated by $Gen$, %(1^{|K|})$,
$D_K(C(\ENC_K(X))=X$ with probability $1 - \delta$.
We call the pair $(C,D)$ an $(\E,\P, \delta)$-{\sf PEC} scheme.

\end{definition}

\smallskip
Note that in this definition, $D$ is given access to the encryption key
while $C$ works independently of the key.
Both $C$ and $D$ may %(and often will)
be built for a specific distribution
$\P$ (in particular, correct decoding may be guaranteed, with high probability,
only for plaintexts chosen according to a specific $\P$).
We often assume that the plaintext
distribution $\P$ is efficiently samplable, which means that there exists an efficient
randomized algorithm whose output distribution is $\P$. The probability of error $\delta$ is 
taken over the choice of $X$ and the choice of random
coins if $D$ is randomized. To simplify notation, we shall
omit the $(\E,\P, \delta)$ parameters in the exposition if they are irrelevant or evident from the
context, and use the term PEC scheme.

%If the $(C,D)$ scheme works for any encryption scheme in a given family $F$ of
%ciphers (e.g., block ciphers in ECB mode, or on CBC mode, stream ciphers, etc)
%we call $(C,D)$ and $(F,\X)$-PEC scheme.
PEC schemes may be tailored to a specific cipher, say AES or DES.
Most often, however, one is interested in schemes that can support different
ciphers, for instance both AES and DES, or even a full family of encryption schemes.
All PEC schemes presented in this paper can
work with any block cipher, therefore we call them {\sf generic} PEC schemes.
A generic PEC scheme cannot be tailored to the specific details of
the underlying cipher but rather handles the cipher as a black box.
That is, it is not necessary that the decoder $D$ knows the
specifics of the encryption scheme or even its
encryption/decryption key. Rather, it suffices that $D$ has access to a pair
of encryption and decryption oracles, denoted by $\ENC$ and $\DEC$, that provide
$D$ with encryptions and decryptions, respectively, of any plaintext or
ciphertext queried by $D$.

\section{Compressing Block-cipher Encryption}
\label{sec:blockCiphers}

In contrast to stream ciphers, such as the one-time pad, block ciphers are highly
nonlinear and the correlation between the key and the
ciphertext is, by design, hard to characterize. If a block cipher operates on each block of data individually, two
identical inputs will produce two identical outputs. While this weakness
does not necessarily enable an unauthorized user to understand 
contents of an individual block it can reveal valuable information; for example, about
frequently occurring data patterns. To address this problem, various
chaining modes, also called modes of operation, are used in conjunction
with block ciphers. The idea is to randomize each plaintext block,
by using a randomization vector derived as a function of previous
encryptor inputs or outputs. The randomization prevents two identical
plaintext blocks from being encrypted into two identical ciphertext blocks, thus preventing 
leakage of information about data patterns. 

We are interested in the following problem. Consider a sequence of plaintext blocks $\mathbf{X}^n=\{X_i\}_{i=1}^n$, where 
each block %$X_i$ consists of $m$ bits; i.e., each plaintext block 
$X_i$ is drawn from the set ${\mathcal{{X}}}^{m} = \{0,1\}^m$. Further,
we assume that the blocks $X_i$ are generated by an i.i.d. source with a distribution $P_X$.
The blocks in $\mathbf{X}^n$ are encrypted with a block-cipher based private-key encryption scheme $(Gen,\ENC,\DEC)$.
In most cases of interest, block-cipher based encryption schemes  use an initialization vector IV that is drawn uniformly at 
random from ${\mathcal{{X}}}^{m}$ by the encryption algorithm $\ENC_K$. Let the encryption algorithm be characterized
by the mapping $\ENC_K: (\mathcal{X}^m)^n \rightarrow \mathcal{X}^m \times (\mathcal{X}^m)^n $.
For a sequence of plaintext blocks $\mathbf{X}^n$ at input, the encryption algorithm generates $\ENC_K(\mathbf{X}^n) = \{\text{IV}, \mathbf{Y}^n\}$, where 
$\mathbf{Y}^n = \{Y_i\}_{i=1}^n$ denotes a sequence of ciphertext blocks and each block $Y_i \in {\mathcal{{X}}}^{m}$.
The problem at hand is to compress $\ENC_K(\mathbf{X}^n)$ without knowledge of $K$.

In the remainder of this section we focus on the cipher block chaining (CBC) mode and the electronic code book (ECB) mode. 
The CBC mode is interesting because it is the most common mode of operation used with block ciphers (e.g. in Internet protocols TLS and IPsec), while our treatment of the ECB mode provides 
fundamental insight
about the feasibility of performing compression on data compressed with block ciphers without chaining. Two other modes of operation associated with
block ciphers are worth mentioning: the output feedback (OFB) mode  and the cipher feedback (CFB) mode. The solution to compressing the latter two 
modes is a relatively straightforward extension of the methods from~\cite{Johnson04OnCompressing}, therefore it is presented in Appendix~\ref{app:ofb_cfb}.

\subsection{Cipher Block Chaining (CBC)}
\label{ssec:cbc}

The most common mode of operation is CBC. Depicted in Figure~\ref{fig:cbc}, block ciphers in CBC mode are employed 
 as the default mechanism in widespread security standards such as 
 IPSec~\cite{Kent05Security} and TLS/SSL~\cite{Dierks08TLS} and hence it is a common method of encrypting internet traffic.

% \ignorefigure{
\begin{figure}[htb!]
\centering
\includegraphics[width=0.65\textwidth]{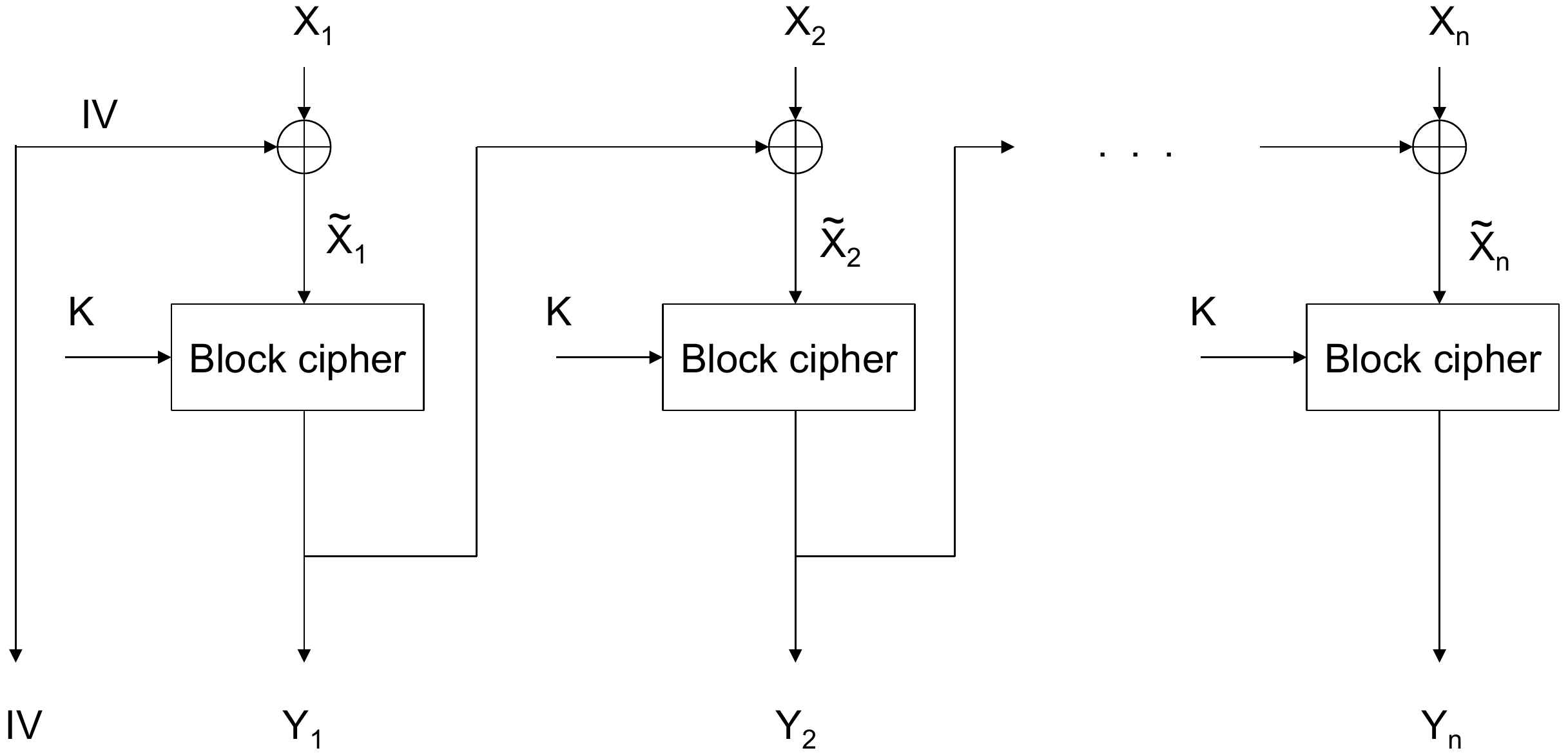}
\caption{Cipher block chaining (CBC).}
\label{fig:cbc}
\end{figure}
% }%ignorefigure

In the CBC mode, each plaintext
block $X_{i}$ is randomized prior to encryption, by being XOR-ed with
the ciphertext block corresponding to the previous plaintext block $Y_{i-1}$, to obtain $\tilde{X}_i$. 
Ciphertext block $Y_i$ is
generated by applying the block cipher
with key $K$ to the randomized plaintext block $\tilde{X}_i$

\begin{equation}
Y_i=B_K\big(X_i\oplus Y_{i-1}\big),
\label{eq:CBC}
\end{equation}
\noindent 
where $Y_0 = \text{IV}$ and $B_K: {\mathcal{{X}}}^{m} \rightarrow {\mathcal{{X}}}^{m}$ is the block cipher mapping using the key $K$. At the output of the encryption algorithm we have $\ENC_K(\mathbf{X}^n) = (\text{IV}, \mathbf{Y}^n)$.

Notice that, contrary to modes OFB and CFB, ciphertext blocks in $\mathbf{Y}^n$ are not obtained by means
of a bitwise XOR operation. Instead, they are obtained as outputs of highly nonlinear block ciphers,
 therefore the methods from~\cite{Johnson04OnCompressing} cannot be applied directly to compress in CBC mode.

The key insight underlying the proposed approach for compression can now be described.
The statistical relationship between the key $K$ and
the $i$-th ciphertext block $Y_i$
is hard to characterize. However, the joint distribution of the randomization
vector $Y_{i-1}$ and the $i$-th input to the block cipher, $\tilde{X}_i$,
is easier to characterize, as it is governed by the distribution of
the plaintext block $X_i$. For example, in the i.i.d source case we
are considering, $Y_{i-1}$ and $\tilde{X}_{i}$ are
related through a symbol-wise model governed by the distribution $P_{X}$.
The correlation induced by the use of the chaining mode can be exploited
to allow compression of encrypted data using Slepian-Wolf coding,
as we will now show.

Let $(C_\text{CBC},D_\text{CBC})$ denote a  Slepian-Wolf code
with encoding rate $R$ and block length $m$. The Slepian-Wolf encoding function
is defined as $C_\text{CBC}: \mathcal{X}^m \rightarrow \{1,\ldots,2^{mR}\}$, and the Slepian-Wolf decoding function as
$D_\text{CBC}:  \{1,\ldots,2^{mR}\} \times \mathcal{X}^m \rightarrow \mathcal{X}^m$.
The proposed compression method is illustrated
in Figure \ref{fig:compressor}. 
% \ignorefigure{
\begin{figure}[htb!]
\centering
\includegraphics[width=0.6\textwidth]{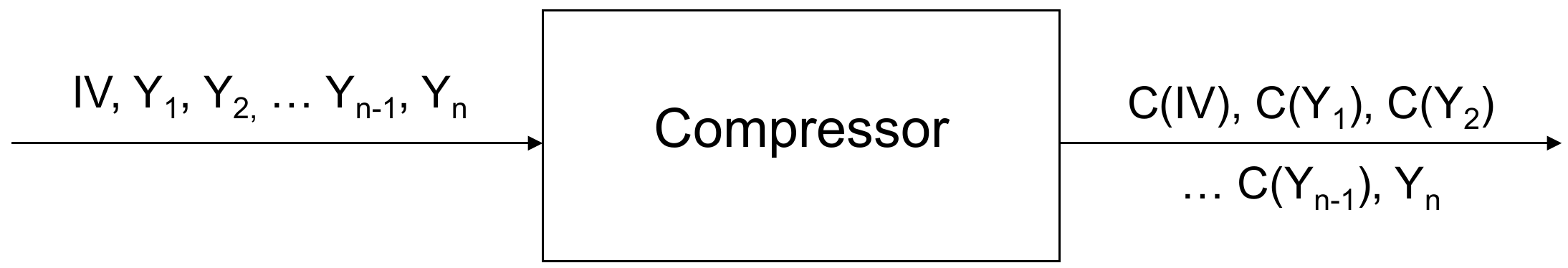}
\caption{Compressor.}
\label{fig:compressor}
\end{figure}
% }%ignorefigure
The input to the compressor is $\ENC_K(\mathbf{X}^n) = (\text{IV}, \mathbf{Y}^n)$. Since $Y_i\in{\mathcal{{X}}}^{m}$,
the total length of the input sequence is $(n+1)\cdot m\cdot\log|{\mathcal{{X}}}|$
bits. The compressor applies the Slepian-Wolf encoding function $C_\text{CBC}$ to the IV and
each of the first $n-1$ ciphertext blocks independently, while the
$n$-th block is left unchanged. Thus, the output of the compressor
is the sequence $(C_\text{CBC}(\text{IV}), C_\text{CBC}(Y_1),\ldots C_\text{CBC}(Y_{n-1}),Y_n)$. The length
of the output sequence is $n\cdot m\cdot R+m\cdot\log|{\mathcal{{X}}}|$
bits. Thus, the compressor achieves a compression factor of

\begin{displaymath}
\lim_{n\rightarrow\infty} \frac{(n+1)\cdot m\cdot\log|{\mathcal{{X}}}|}{n\cdot m\cdot R+m\cdot\log|{\mathcal{{X}}}|} = \frac{\log|{\mathcal{{X}}}|}{R}
\end{displaymath}

\noindent 
for large $n$. Note that the compressor does not need to know the key $K$. Also, note that this approach only requires a compressed IV, which by itself is incompressible, 
therefore no performance loss is inflicted by the uncompressed last block.

The joint decompression and decryption method is shown in Figure \ref{fig:decompressor}.
The received compressed sequence is decrypted and decompressed serially, from right to left.
In the first step $Y_n$, which is received uncompressed, is decrypted using the key $K$ to generate $\tilde{X}_{n}$.
% \ignorefigure{
\begin{figure}[htb!]
\centering
\includegraphics[width=0.65\textwidth]{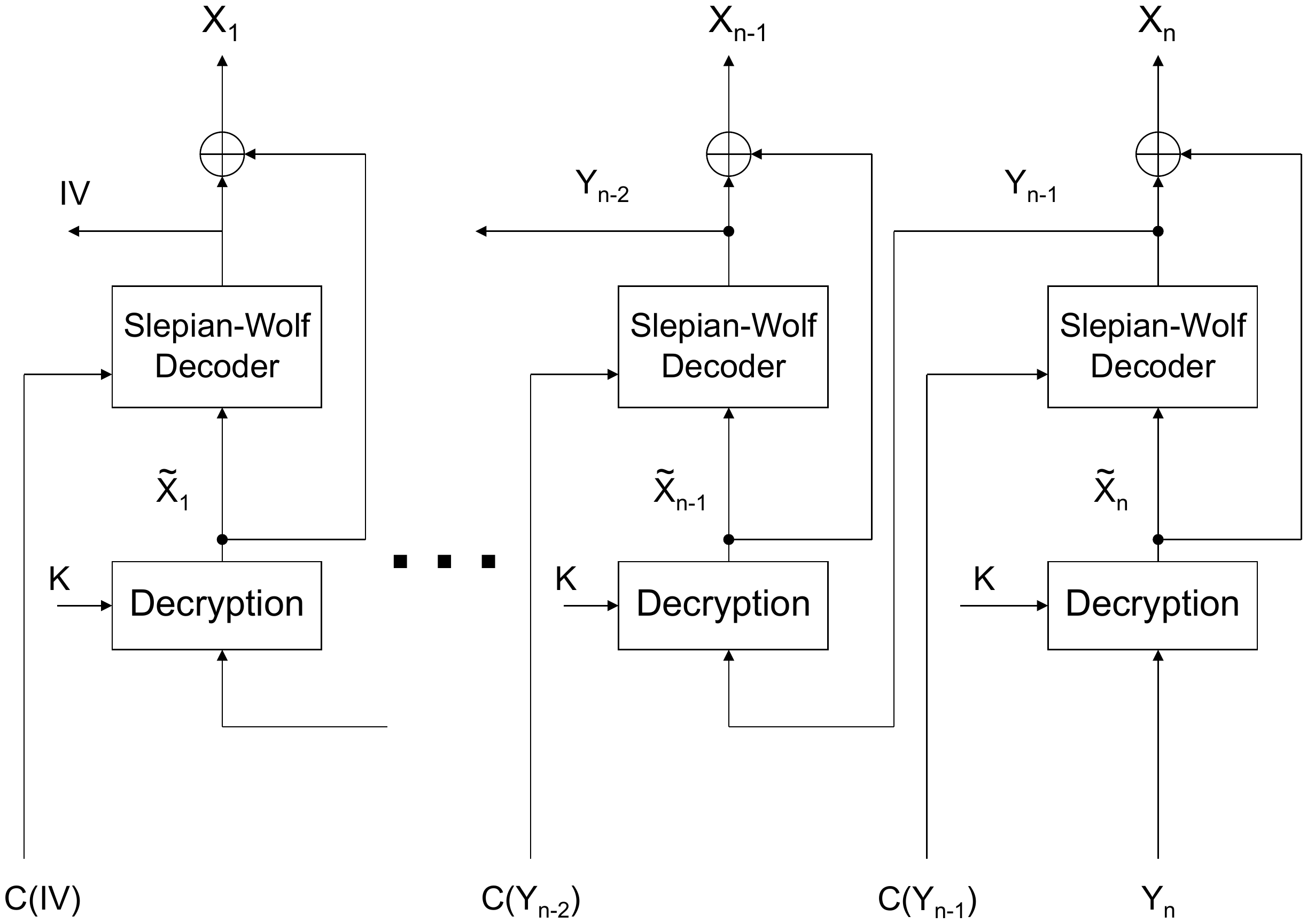}
\caption{Joint decryption and decoding at the receiver. It is performed serially from right to left.}
\label{fig:decompressor}
\end{figure}
% }%ignorefigure
Next, Slepian-Wolf decoding is performed to reconstruct $Y_{n-1}$
using $\tilde{X}_{n}$ as side-information, and the compressed bits
$C_\text{CBC}(Y_{n-1})$. The decoder computes $\hat{Y}\triangleq D_\text{CBC}(C_\text{CBC}(Y_{n-1}),\tilde{X}_{n})$,
such that $\hat{Y}=Y_{n-1}$ with high-probability
if the rate $R$ is high enough. Once $Y_{n-1}$ has been recovered by the Slepian-Wolf
decoder, the plaintext block can now be reconstructed as $X_n=Y_{n-1}\oplus\tilde{X}_{n}$.
The decoding process now proceeds serially with $Y_{n-1}$
decrypted to generate $\tilde{X}_{n-1}$, which acts as the new Slepian-Wolf
side-information. This continues until all plaintext blocks have been
reconstructed. 

For large $m$, it follows from the
Slepian-Wolf theorem that the rate required to ensure correct reconstruction
of the $(i-1)$-th block with high probability is given as

\begin{eqnarray}
mR & = &H\big(Y_{i-1}|\tilde{X}_i\big) = H\big(Y_{i-1}|Y_{i-1}\oplus X_i\big) \nonumber \\
  & = &H\big(Y_{i-1}, Y_{i-1}\oplus X_i |Y_{i-1}\oplus X_i\big) = H\big(Y_{i-1}, X_i |Y_{i-1}\oplus X_i\big) \nonumber \\
  & = &H\big(X_i | Y_{i-1} \oplus X_i\big) \leq H(X_i).\label{eq:R_sw}
\end{eqnarray}
\noindent  
It is assumed that the IV is drawn uniformly at random, therefore $Y_i$ is uniformly distributed for all $i$. Consequently,
the equation~\ref{eq:R_sw} holds with equality and we have $R = \frac{1}{m} H(X_i)$.
%From the source coding theorem we also have that $R \geq H(X_i)$, hence it follows that $R = H(X_i)$.

In practice, as we will see in Section 4, $m$ is typically
small. In this case, the required rate $R$ is a
function of $P_{X}$, $m$, the acceptable decoding error probability,
and the non-ideal Slepian-Wolf codes used.

\subsection{Electronic Code Book (ECB)}
\label{ssec:ecb}

%The first mode we discuss is the electronic code book (ECB) mode, depicted in Figure~\ref{fig:ecb}, where
%each plaintext block is encrypted individually and independent of all others.
We have seen that ciphertexts generated by a block cipher in OFB, CFB, and CBC modes can be compressed without knowledge of the
encryption key. The compression
schemes that we presented rely on the specifics of chaining operations.
A natural question at this point is to what extent can the output of a block
cipher be compressed without chaining, i.e. when a block cipher is applied to a
single block of plaintext. This mode of operation, depicted in Figure~\ref{fig:ecb},
is called the electronic code book (ECB) mode.

% \ignorefigure{
\begin{figure}[htb!]
\centering
\includegraphics[width=0.65\textwidth]{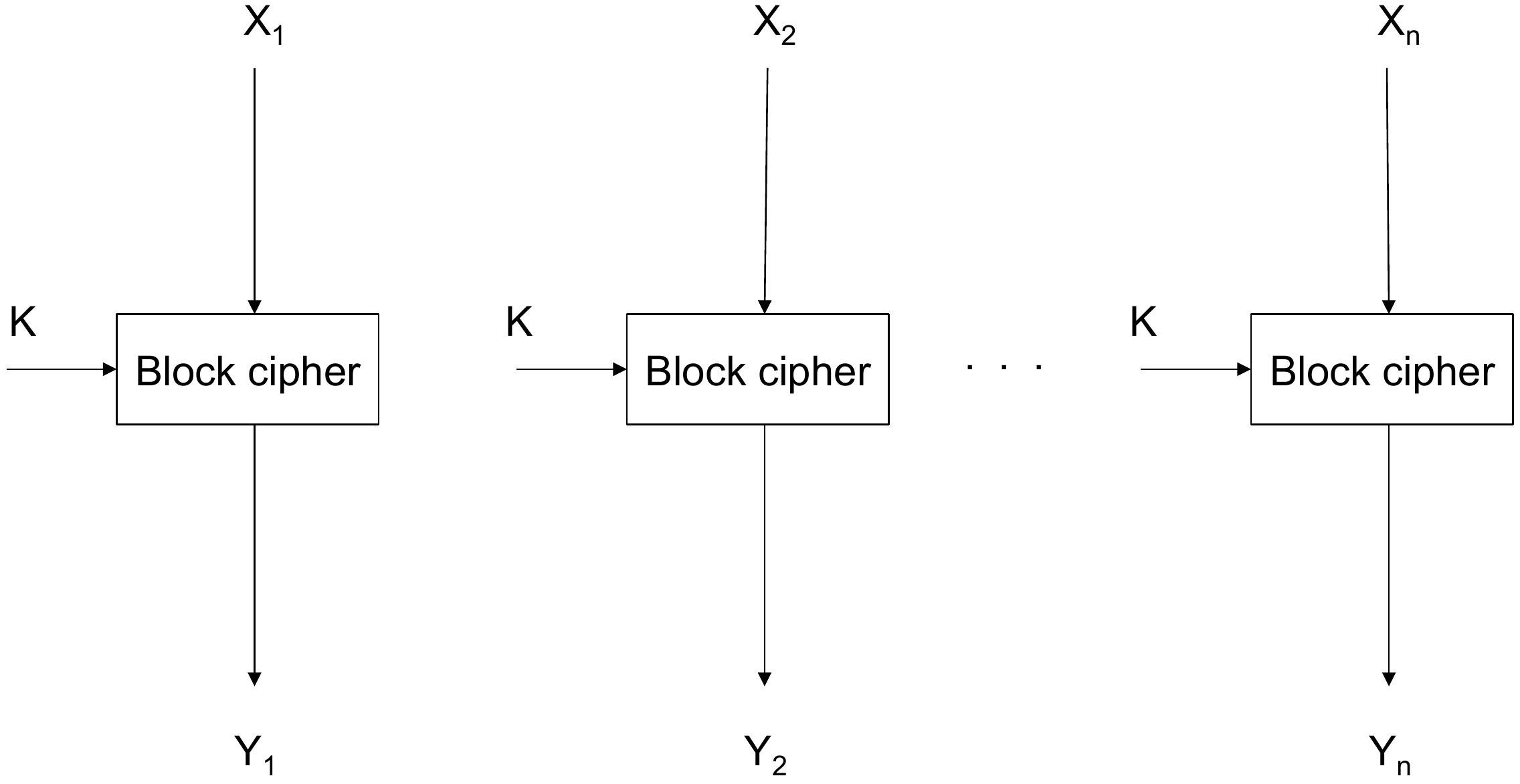}
\caption{Electronic Codebook (ECB).}
\label{fig:ecb}
\end{figure}
% }%ignorefigure

Block cipher with a fixed key is a permutation, therefore the
entropy of the ciphertext is the same as the entropy of the plaintext. In effect, compression in ECB mode is theoretically
possible.
The question is,  whether it is possible to design a {\em generic and efficient}
post-encryption compression scheme such as with other modes of operation.

We claim the answer to this question is negative, except for some low-entropy
distributions or very low compression rates (e.g., compressing a ciphertext
by a few bits).
We show that for a given compressed ciphertext $C(Y_i)$, the decoder cannot do
significantly better than to operate in one of the following two
decoding strategies, which we refer to as {\em exhaustive} strategies:
\begin{enumerate}
\item enumerate all possible plaintexts in decreasing order of
probability and compute a ciphertext for each plaintext until the
ciphertext that compresses to $C(Y_i)$ is found;
\item enumerate all ciphertexts that compress to $C(Y_i)$ and decrypt them to find
the original plaintext.
\end{enumerate}
More precisely, we show that a generic compression scheme that
compresses the output of a block cipher and departs significantly from
one of the above strategies can be converted into an algorithm that
breaks the security of the block cipher. In other words, a scheme
that compresses the output of a secure block cipher either requires an
infeasible amount of computation, i.e. as much as needed to break the block cipher, 
or it must follow one of the above exhaustive strategies. 

%\noindent{\bf Definition.} 
\begin{definition}
If $(C,D)$ is a PEC scheme that follows any of the
above exhaustive strategies, we say that $(C,D)$ is an {\em exhaustive PEC
scheme.}
\end{definition}

Our result on the compressibility of block ciphers in ECB mode can now be stated.

\smallskip

\begin{theorem}
\label{thm:lowerbound}
Let $B$ be a secure block cipher\footnote{For a definition of block cipher security please refer to Appendix~\ref{sec:cryptdefs}.}, let $\E=(Gen,\ENC,\DEC)$ be an encryption scheme, where
$Gen$ chooses $K$ uniformly from $\K$ and $\ENC=B_K$, $\DEC=B_K^{-1}$, and let $\P$ be an efficiently-samplable plaintext
distribution. Let $(C,D)$ be a $(\E,\P, \delta)$-PEC scheme. If $(C,D)$ is generic for block ciphers, then
$(C,D)$ is either exhaustive or computationally infeasible
(or both).%
\end{theorem}

\smallskip

The exact computational bounds are omitted in the Theorem's
statement for simplicity. More details on the computational bounds are given in Appendix~\ref{app:sec-proof} (Lemma 11), where 
Theorem~\ref{thm:lowerbound} is also proven. % (in particular, see Lemma~\ref{lemma-lowerbound}). 
The following remarks are worth noting:

\begin{itemize}
\item[(a)] The exhaustive strategies are infeasible in most cases. However, for very low-entropy plaintext
distributions or for very low compression rates (e.g., when compressing a
ciphertext by just a few bits) these strategies can be efficient.
For example, consider a plaintext distribution consisting of a set $\X$ of
1,000 128-bit values uniformly distributed.  In this case, one can compress
the output of a 128-bit block cipher applied to this set by truncating the
128-bit ciphertext to 40 bits. Using a birthday-type bound one can show that
the probability of two values in $\X$ being mapped to the same ciphertext is
about $2^{-20}$. Hence, exhaustive strategy 1 above would succeed in
recovering the correct plaintext with probability $1-2^{-20}$.
In general, the compression capabilities under this strategy will depend on
the {\em guessing entropy} \cite{Massey94Guessing,Cachin97Entropy}
of the underlying plaintext distribution.
Another compression approach is to drop some of the bits in the ciphertext and
let the decoder search exhaustively, as in the strategy 2 above, for the original
ciphertext until the correct plaintext is recovered. This would assume that
the number of dropped bits allows for almost-unique decodability and that the
decoding procedure can identify the right plaintext.
The fact that these exhaustive compression strategies may be efficient for some
plaintext distributions shows that our theorem cannot exclude the existence of
efficient generic coding schemes for some plaintext distributions. Rather, the
theorem says that if there exists an efficient generic coding scheme for a given
distribution, then it must follow one of the exhaustive strategies.

\item[(b)] Compressing the ciphertext 
is theoretically possible. If
efficiency is not of concern, one could consider a brute-force compression algorithm that
first breaks the block cipher by finding its key\footnote{
This can be done by exhaustive search based on a known plaintext-ciphertext
pair or (for suitable plaintext distributions)
by decrypting a sequence of encrypted blocks and finding a key which
decrypts all blocks to elements in the underlying probability distribution.},
uses the key to decrypt, and then compresses the plaintext. If several blocks
of plaintext are compressed sequentially, they can be re-encrypted by the
compression algorithm resulting in an effective ciphertext compression.
This idea might be impractical, but it shows that if the existence of
generic compressors is ruled out, this must be tied to the efficiency of the PEC scheme.
Moreover, the above approach could work efficiently against an
insecure block cipher (in which case one may find the key by efficient
means), thus showing that {\em both efficiency and security} are essential
ingredients to our result.
%This also shows that it is not possible to show a lower bound on the rate of
%compression of encrypted text but only on the efficiency of the algorithm.

\item[(c)] Theorem~\ref{thm:lowerbound} holds for generic compressors 
 that do not use the internals of an encryption
algorithm or the actual key in the (de)compression process. It does not rule out
the existence of good PEC schemes for specific secure block ciphers, like for 
instance AES. To achieve compression, though, the PEC scheme would have to
be contingent on the internal structure of a block cipher.
%Our results do not preclude the possible existence of a good PEC scheme
%for a specific secure block cipher, say AES.  Rather, these results
%hold for generic compressors that do not use the internals of the encryption
%algorithm or the actual key in the (de)compression process.  {\bf It is doubtful
%that a practical block cipher will be significantly compressible,
%but, in principle, it could be the case even for a secure block cipher -- need to rewrite.}
\end{itemize}

%%%%%%%%%%%%%%%%%%%%%%%%%%%%%%%%%%%%%%%%%%%%%%%%%%%%%%%%%%%%%%%%%%%%%%
\section{Security of the CBC Compression Scheme}
\label{sec:security}
%%%%%%%%%%%%%%%%%%%%%%%%%%%%%%%%%%%%%%%%%%%%%%%%%%%%%%%%%%%%%%%%%%%%%%

In this section we formally prove that compression and decompression
operations that we introduced on top of the regular CBC mode do not compromise
security of the original CBC encryption.  The proof follows standard
techniques in the cryptographic literature, 
%and uses the fact that both compression and decompression algorithms access the
%encryption and decryption functionalities as black boxes (rather than utilizing
%the internals of the block cipher).  On that basis, one can show 
showing that any efficient attack against
secrecy of a PEC scheme can be transformed
into an efficient attack against the original CBC encryption. Please refer to Appendix~\ref{sec:cryptdefs} for background 
on cryptographic definitions that are used in this section.
% (here we also use the fact that our compression and decompression operations
% are in themselves efficient).

We start by formalizing the notion of security of a PEC scheme as a simple
extension of the standard definition of chosen plaintext attack (CPA) security recalled in Appendix~\ref{sec:cryptdefs}.
The essence of the extension is that in the PEC setting
the adversary
is given access to a combined oracle $(\ENC_K+C)(\cdot)$ which first encrypts the
plaintext  and then compresses the resultant ciphertext. 

\smallskip
\begin{definition}[CPA-PEC Indistinguishability Experiment]
Consider an encryption scheme $(Gen,\ENC,\\\DEC)$ and a PEC scheme $(C,D)$.  
The {\em CPA-PEC indistinguishability experiment} ${\sf Expt_\A^{cpa-pec}}$ is defined as follows:
\begin{enumerate}
\item a key $K$ is generated by running $Gen$;

\item the adversary $\A$ has oracle access to $(\ENC_K+C)(\cdot)$, and queries
it with a pair of test plaintexts $X_0,X_1$ of the same length;

\item  the oracle randomly chooses a bit $b\leftarrow\{0,1\}$ and returns the compressed ciphertext $c\leftarrow(\ENC_K+C)(X_b)$, called the challenge,  to $\A$;

\item the adversary $\A$ continues to have oracle access to $(\ENC_K+C)(\cdot)$ and is allowed to make arbitrary queries. Ultimately it makes a guess about the value of $b$ by outputting $b'$;

\item the output of the experiment, ${\sf Expt_\A^{cpa-pec}}$, is defined to be $1$ if $b' = b$, and $0$ otherwise. If ${\sf Expt_\A^{cpa-pec}}=1$, we say that $\A$ succeeded.

\end{enumerate}
\end{definition}

\medskip

\begin{definition}[Post-Encryption Compression Security]\label{def-cpa-pec}
A PEC scheme $(C,D)$ is called  $(T,\eps)$-indistin\-guish\-able under chosen
plaintext attacks ({\em CPA-PEC-secure}) if for every adversary $\A$
that runs in time $T$, 
$$
\Pr\Big[{\sf Expt_\A^{cpa-pec}} = 1\Big] < \frac{1}{2} + \eps,
$$
where the probability is taken over all random coins used by $\A$, as well as all random coins used in the experiment.
\end{definition}

\smallskip

We now formulate the security of our CBC compression scheme by the following theorem.
%The proof relies on the result by Bellare et al.~\cite{Bellare97Concrete} who proved that if $B$ is a secure block cipher then the CBC encryption mode for $B$ is a secure CPA encryption scheme. 

\begin{theorem}
Let $\E=(Gen,\ENC,\DEC)$ be a CBC encryption scheme that is $(T,\eps)$-indistinguishable under chosen plaintext attacks, let $\P$ be an efficiently-samplable plaintext distribution, and let $(C,D)$ be a $(\E,\P,\delta)$-PEC scheme. Then $(C,D)$ is $(T/ T_C,\eps)$-indistinguishable under chosen plaintext attacks, where $T_C$ is an upper bound on the running time of $C$.
\end{theorem}

\begin{proof}
Our proof employs a reduction to the security of $\E$. Specifically, we show that if there exists an adversary $\A_C$ that runs in time $T/ T_C$ and is able to distinguish between two compressed encryptions,
%\footnote{Recall that we consider a game where $\A_C$ is given oracle access to $(\ENC_K+C)(\cdot)$ (which first encrypts and then compresses). $\A_C$ is allowed to query its oracle at any point with $X$, in response to which the oracle returns $C(\ENC_K(X))$. Finally, $\A_C$ outputs two plaintexts $X_0,X_1$, and receives back a compressed ciphertext $c=C(\ENC(X_b))$, where $b\in_R\{0,1\}$. $\A_C$ is allowed to query its oracle after receiving its challenge as well. The game is concluded when $\A_C$ outputs a bit $b'$ and we say that it distinguishes successfully in this game if $b'=b$.} 
then there exists an adversary $A_\E$ that runs in time $T$ and compromises the security of $\E$. 
Note that the latter implies a break of security of the underlying block cipher using the well-known result by Bellare et al.~\cite{Bellare97Concrete}.

Assume, for contradiction, the existence of such an adversary $\A_C$. We construct an adversary $\A_\E$ as follows. $\A_\E$ invokes $\A_C$ and emulates its oracle $(\ENC_K+C)(\cdot)$. That is, for every query $X$ made by $\A_C$, $\A_\E$ uses its oracle $\ENC_K(\cdot)$ to compute $\ENC_K(X)$. Subsequently, it applies the compression algorithm $C$ on $\ENC_K(X)$  and forwards $C(\ENC_K(X))$ to $\A_C$. 
When $\A_C$ outputs the two test messages $X_0,X_1$, $\A_\E$ outputs these messages to its own oracle $\ENC_K(\cdot)$. Let 
$c^*$ denote the challenge ciphertext that its oracle returns. $\A_\E$ computes $c=C(c^*)$ and forwards $c$ to $\A_C$. 
 When $\A_C$ outputs a bit $b'$, $\A_\E$ outputs the same value. 

Given that $\A_C$ makes at most $T/ T_C$ queries and that the running time of $C$ is upper-bounded by $T_C$, $\A_\E$'s running time is at most $T$. Let $\frac{1}{2}+\epsilon$ denote the probability that $\A_C$ distinguishes successfully in its game. Clearly, it holds that $\A_\E$ distinguishes successfully in its game with the same probability, which is in contradiction to the security of $\E$.

\end{proof}

%%%%%%%%%%%%%%%%%%%%%%%%%%%%%%%%%%%%%%%%%%%%%%%%%%%%%%%%%%%%%%%%%%%%%%
%\section{Compressing Public-key Encryption schemes}
%%%%%%%%%%%%%%%%%%%%%%%%%%%%%%%%%%%%%%%%%%%%%%%%%%%%%%%%%%%%%%%%%%%%%%

\section{Public-key Encryption Schemes}
\label{sec:publicEncryptionSchemes}

After a feasibility study of post-encryption compression
in the symmtric key setting, where sender and receiver share a key used for both
encryption and decryption,  one may wonder whether these results can be extended
to the public-key setting. 
Loosely speaking, in the public-key approach the key generation algorithm
produces two different keys (PK,SK) where PK is publicly distributed, yet SK is
kept a secret. Thereafter, any party who knows PK can encrypt messages and only
the owner of the secret key can decrypt.

Since public key operations for known public-key schemes are computationally
expensive, one typically encrypts streams of data by using a public-key scheme to
encrypt a symmetric key (e.g., a AES key) and then  using this key with a
symmetric cipher to encrypt the data. In this case, it
is clear that the results in this paper can be applied to the symmetric (bulk)
encryption of the data.
Yet, the question remains of whether we can compress the PK part of the
ciphertext.
For the reader familiar with the El Gamal cryptosystem ~\cite{ElGamal84Public},
we note that one can compress one of the two components of the
ciphertext if, for example, one uses XOR to combine the ephemeral El-Gamal key
with the message.
In the case of El Gamal over elliptic curves another well-known technique is
``point compression" \cite{Blake00Eliptic} which reduces the ciphertext to half its length.
Yet, one could hope for better post-encryption compression (in particular,
given that often the ciphertexts are longer than the plaintext). 
One result in this direction is Gentry's
technique for compressing Rabin's ciphertexts ~\cite{Gentry04How}. However, this technique
can be applied {\em before} encryption, not post-encryption.
Interestingly, the related result by Gentry regarding compression of Rabin's
signatures does allow for compression {\em after} the signature generation,
without the need for the signing key.

In conclusion, the question of post-encryption compressibility of ciphertexts
produced by public-key schemes is essentially unresolved.

\section{Compression Performance}
\label{sec:simResults}

Codes that are used to compress stream ciphers can be chosen to have arbitrary block lengths, since the method itself 
does not directly impose any constraints on the block length. On the other hand, the compression methods that we propose for CFB and CBC mode  must operate block-wise since decompression occurs serially and depends upon a previously decompressed block. In effect, the block length of Slepian-Wolf codes must be $m$.  We choose AES as a representative of a widely used block cipher and present our results under assumption that AES is used as the encryption scheme. 

The efficiency of Slepian-Wolf compression depends on the performance of underlying Slepian-Wolf codes. It was shown in~\cite{He06ALower, He06OnRelationship}
that Slepian-Wolf compression approaches entropy with speed $O(\sqrt{\frac{\log{n}}{n}})$, which is considerably slower than $O(\frac{1}{n})$ of arithmetic coding. It follows that for efficient Slepian-Wolf compression, the block length of Slepian-Wolf codes must be long.

Slepian-Wolf codes over finite block lengths have nonzero frame-error rates (FER), which implies that the receiver will sometimes fail to recover $E_K(X_{i-1})$ correctly. Such errors must be dealt with on the system level, as they can have catastrophic consequences in the sense that they propagate to all subsequent blocks. In the following it is assumed that as long as the FER is low enough, the system can recover efficiently, for instance by supplying the uncompressed version of the erroneous block to the receiver. The compression performance depends on the target FER.

We consider a binary i.i.d source with a probability distribution $\Pr(X = 1) = p$ and $\Pr(X=0) = 1-p$ which produces plaintext bits. A sequence of plaintext bits is divided into blocks of size $m$, which are encrypted and compressed as described in Section~\ref{sec:blockCiphers}. The receiver's task is to reconstruct $E_K(X_{i-1})$ from $\tilde{X}_i$ and side information $C(E_K(X_{i-1}))$. For the source considered, Slepian-Wolf decoding  is equivalent to error-correction over a binary symmetric channel (BSC), thus the underlying codes should yield a FER that is lower or equal to the target FER over the BSC. Compression efficiency can be evaluated in two ways:

\begin{itemize}
\item[a)] fix $p$ and determine the compression rate of a Slepian-Wolf code that satisfies the target FER;
\item[b)] pick a well-performing Slepian-Wolf code and determine the maximum $p$ for which target FER is satisfied.
\end{itemize}
We use low-density parity-check (LDPC) codes~\cite{Richardson08Modern}, which are known to be very powerful and we evaluate compression performance according to method b). In our simulations we used two LDPC codes\footnote{The degree distributions were obtained at http://lthcwww.epfl.ch/research/ldpcopt/}: the first yields compression rate 0.5 and has degree distribution $\lambda(x) = 0.3317x + 0.2376x^2 + 0.4307x^5$, $\rho(x)=0.6535x^5 + 0.3465x^6$ and the second yields compression rate 0.75 and degree distribution $\lambda(x) = 0.4249x + 0.0311x^2 + 0.5440x^4$, $\rho(x)=0.8187x^3 + 0.1813x^4$. All codes were constructed with the Progressive Edge Growth algorithm~\cite{Hu05Regular}, which is known to yield good performance at short block lengths.  Belief propagation~\cite{Richardson08Modern} is used for decoding and the maximum number of iterations was set to 100. 

\begin{table}[h!]
\centering
\begin{tabular}{|c|c|c|c|}
	\hline
	$p$    & Target FER & Compression Rate & Source Entropy  \\ \hline \hline
	0.026 & $10^{-3}$      & 0.50                         &  0.1739 \\ \hline
	0.018 & $10^{-4}$      & 0.50                         &  0.1301 \\ \hline
	0.068 & $10^{-3}$      & 0.75                         &  0.3584 \\ \hline
	0.054 & $10^{-4}$      & 0.75                         &  0.3032 \\ \hline
\end{tabular}
\caption{Attainable compression rates for $m= 128$ bits.}
\label{tab:cmprEffic_128}
\end{table}

Simulation results for block length, $m$, 128 and 1024 bits are shown in Tables~\ref{tab:cmprEffic_128} and~\ref{tab:cmprEffic_1024}, respectively. First, consider the current specification of the AES standard~\cite{Mao03Modern}, where  $m$ is 128 bits. At FER of $10^{-3}$ the maximum $p$ that can be compressed to rate 0.5 is 0.026, while the entropy for that $p$ is 0.1739. The large gap is a consequence of a very short block length. Notice that for higher reliability, i.e. lower FER, the constraints on the source are more stringent for a fixed compression rate.
 
%\begin{table}[h!]
%\centering
%\begin{tabular}{|c|c|c|c|}
%	\hline
%	$p$    & Target FER & Compression Rate & Source Entropy  \\ \hline \hline
%	0.039 & $10^{-3}$      & 0.50                         &  0.2377 \\ \hline
%	0.029 & $10^{-4}$      & 0.50                         &  0.1894 \\ \hline
%	0.091 & $10^{-3}$      & 0.75                         &  0.4398 \\ \hline
%	0.076 & $10^{-4}$      & 0.75                         &  0.3879 \\ \hline
%\end{tabular}
%\caption{Attainable compression rates for $m= 256$ bits.}
%\label{tab:cmprEffic_256}
%\end{table}
\begin{table}[h!]
\centering
\begin{tabular}{|c|c|c|c|}
	\hline
	$p$    & Target FER & Compression Rate & Source Entropy  \\ \hline \hline
	0.058 & $10^{-3}$      & 0.50                         &  0.3195 \\ \hline
	0.048 & $10^{-4}$      & 0.50                         &  0.2778 \\ \hline
	0.134 & $10^{-3}$      & 0.75                         &  0.5710 \\ \hline
	0.126 & $10^{-4}$      & 0.75                         &  0.5464 \\ \hline
\end{tabular}
\caption{Attainable compression rates for $m= 1024$ bits.}
\label{tab:cmprEffic_1024}
\end{table}

When the $m$ is increased to 1024 bits (see Table~\ref{tab:cmprEffic_1024}) the improvement in performance is considerable. For instance, at FER $= 10^{-3}$ and compression rate 0.5, the source can now have $p$ up to 0.058. In future block-cipher designs one could consider larger
blocks, in particular as a way to allow for better post-encryption compression.

%\section{System Aspects}

%%Due to additional decoding, the decryption can not be fully parallelized, but if serial decoding is tolerable, compression can be achieved with
%%block ciphers. We note that very similar principles can be used to compress block ciphers that are used in other modes of operations which
%%use XOR to concatenate blocks. In the next section we investigate code design for compression of block ciphers.

%Talk about some system aspects: parallel decoding is not possible in this scheme. Needs to be addressed: error detection, error handling, protocols to combat that ... Error propagation, need to choose an appropriate "super" block size .....trade of between : impact of the uncompressed block vs. the propagation of errors (\# of retransmissions needed in cases of errors)

\section{Conclusion}
\label{sec:conclusion}

We considered compression of data encrypted with block ciphers without knowledge of the key. Contrary to a popular belief that such data is practically incompressible we show how compression can be attained. Our method is based on Slepian-Wolf coding and hinges on the fact that chaining modes, which are widely used in conjunction with block ciphers, introduce a simple symbol-wise correlation between successive blocks of data. 
The proposed compression was shown to preserve the security of the encryption scheme. Further, we showed the existence of a fundamental limitation to compressibility of data encrypted with block ciphers when
no chaining mode is employed.

Some simulation results are presented for binary memoryless sources. The results indicate that, while still far from theoretical limits, considerable compression gains are practically attainable with block ciphers and improved performance can be expected as block sizes increase in the future.

\appendices

\section{Compressing OFB and CFB Mode}
\label{app:ofb_cfb}

The main part of the paper describes how compression can be performed on data encrypted with block ciphers when 
they operate in the CBC mode. This section outlines how compression is achieved in other two common modes
associated with block ciphers,  output feedback (OFB) and cipher feedback (CFB). 

\subsection{Output Feedback (OFB)}
\label{ssec:ofb}

The mode depicted in Figure~\ref{fig:ofb} is called output feedback (OFB). Plaintext blocks in $\mathbf{X}^n$ are not 
directly encrypted with a block cipher. Rather, the block cipher is used to sequentially generate a sequence of pseudorandom
blocks $\mathbf{\tilde{K}}^n = \{\tilde{K}_i\}_{i=1}^n$ which serve as a one-time pad to encrypt the plaintext blocks. At the output of the encryption algorithm we have $\ENC_K(\mathbf{X}^n) = (\text{IV}, \mathbf{Y}^n)$.

% \ignorefigure{
\begin{figure}[htb!]
\centering
\includegraphics[width=0.65\textwidth]{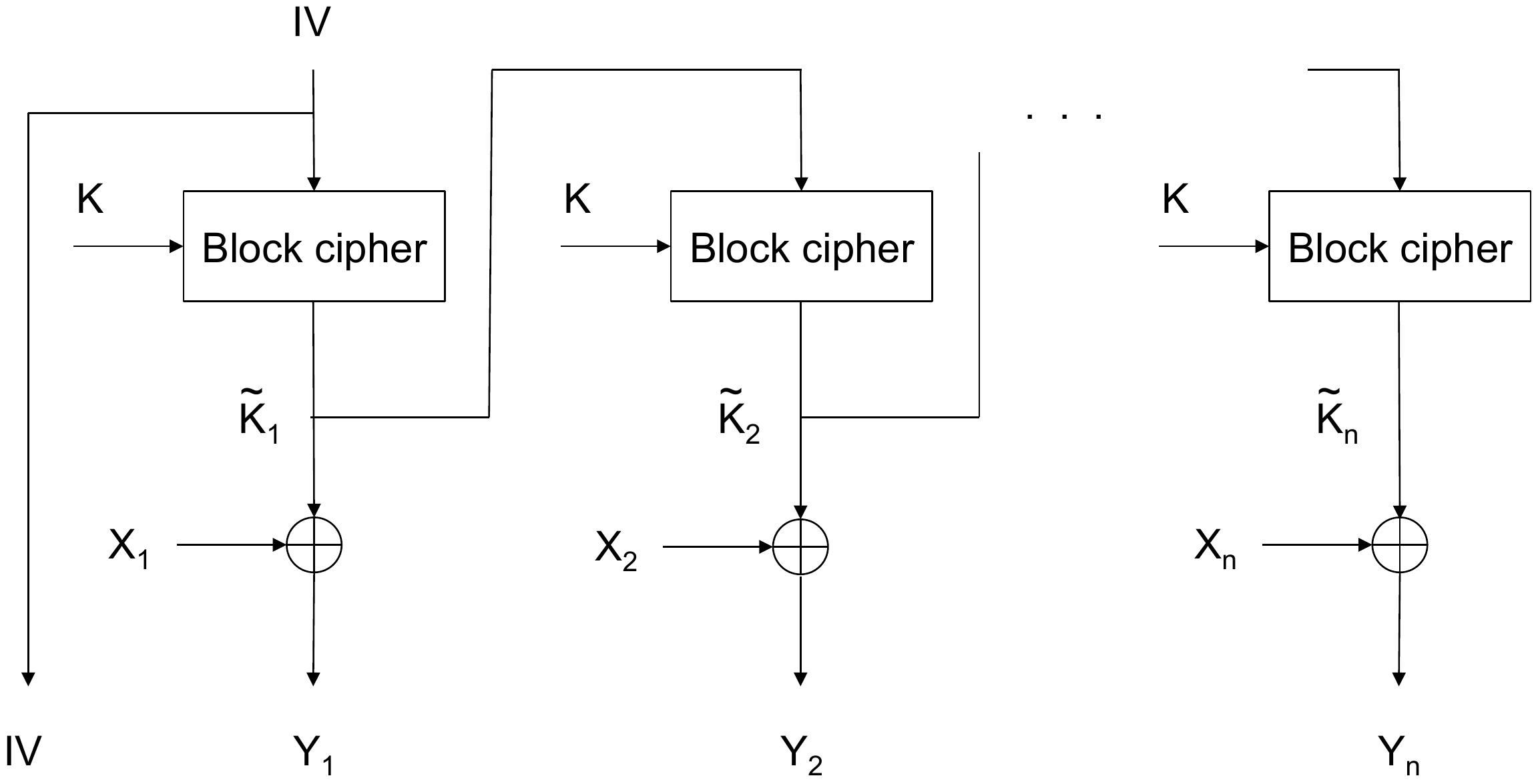}
\caption{Output feedback (OFB).}
\label{fig:ofb}
\end{figure}
% }%ignorefigure

Notice that each block $\tilde{K}_i$  is independent of plaintext blocks, therefore the OFB mode is analogous 
to the one-time pad encryption scheme and $\mathbf{Y}^n$ can be compressed in the same manner as described 
in~\cite{Johnson04OnCompressing}.  Most importantly, $\mathbf{\tilde{K}}^n$ and $\mathbf{X}^n$ are independent, therefore the block length
of Slepian-Wolf codes that are used to compress $\mathbf{Y}^n$ can be chosen arbitrarily. The IV is uniformly distributed and thus incompressible.

Formally, the compression and decompression algorithms can be described as follows. 
Let $(C_\text{OFB},D_\text{OFB})$ denote a  Slepian-Wolf code 
with encoding rate $R$ and block length $nm$. The Slepian-Wolf encoding function
is defined as $C_\text{OFB}:  (\mathcal{X}^m)^n \rightarrow \{1,\ldots,2^{mnR}\}$, and the Slepian-Wolf decoding function as
$D_\text{OFB}:  \{1,\ldots,2^{mnR}\} \times (\mathcal{X}^m)^n \rightarrow (\mathcal{X}^m)^n$. 
Given an output sequence from the encryptor $\ENC_K(\mathbf{X}^n) = (\text{IV}, \mathbf{Y}^n)$, compression is achieved
by applying the Slepian-Wolf encoder function $C_\text{OFB}$ to $\mathbf{Y}^n$, so that at the output we 
have $(\text{IV}, C_\text{OFB}(\mathbf{Y}^n))$. Note that $C_\text{OFB}$ does not require
knowledge of $K$.

Decompression and decryption are performed jointly. The receiver has $(\text{IV}, C_\text{OFB}(\mathbf{Y}^n))$ and knows
the secret key $K$, thus it can generate the sequence of pseudorandom blocks $\mathbf{\tilde{K}}^n$ using the IV. Subsequently,
it applies the Slepian-Wolf decoder function $D_\text{OFB}$ to $(C_\text{OFB}(\mathbf{Y}^n), \mathbf{\tilde{K}}^n)$ to recover $\mathbf{Y}^n$.
The original sequence of plaintext blocks $\mathbf{X}^n$ then equals $\mathbf{\tilde{K}}^n \oplus \mathbf{Y}^n$.

By the Slepian-Wolf theorem, the compression rate approaches
entropy of the source asymptotically in $nm$. That is,  even if $m$ is finite, entropy can still be achieved if $n \rightarrow \infty$.

\subsection{Cipher Feedback (CFB)}
\label{ssec:cfb}

Next, we discuss the cipher feedback (CFB) mode, which is depicted in
Figure~\ref{fig:cfb}. Similarly as in OFB, the plaintext blocks are not
subject to block cipher encryption and $\mathbf{Y}^n$ is obtained by XOR-ing plaintext blocks $\mathbf{X}^n$  with
pseudorandom blocks $\mathbf{\tilde{K}}^n$. At the output of the encryption algorithm we have $\ENC_K(\mathbf{X}^n) = (\text{IV}, \mathbf{Y}^n)$.

% \ignorefigure{
\begin{figure}[htb!]
\centering
\includegraphics[width=0.65\textwidth]{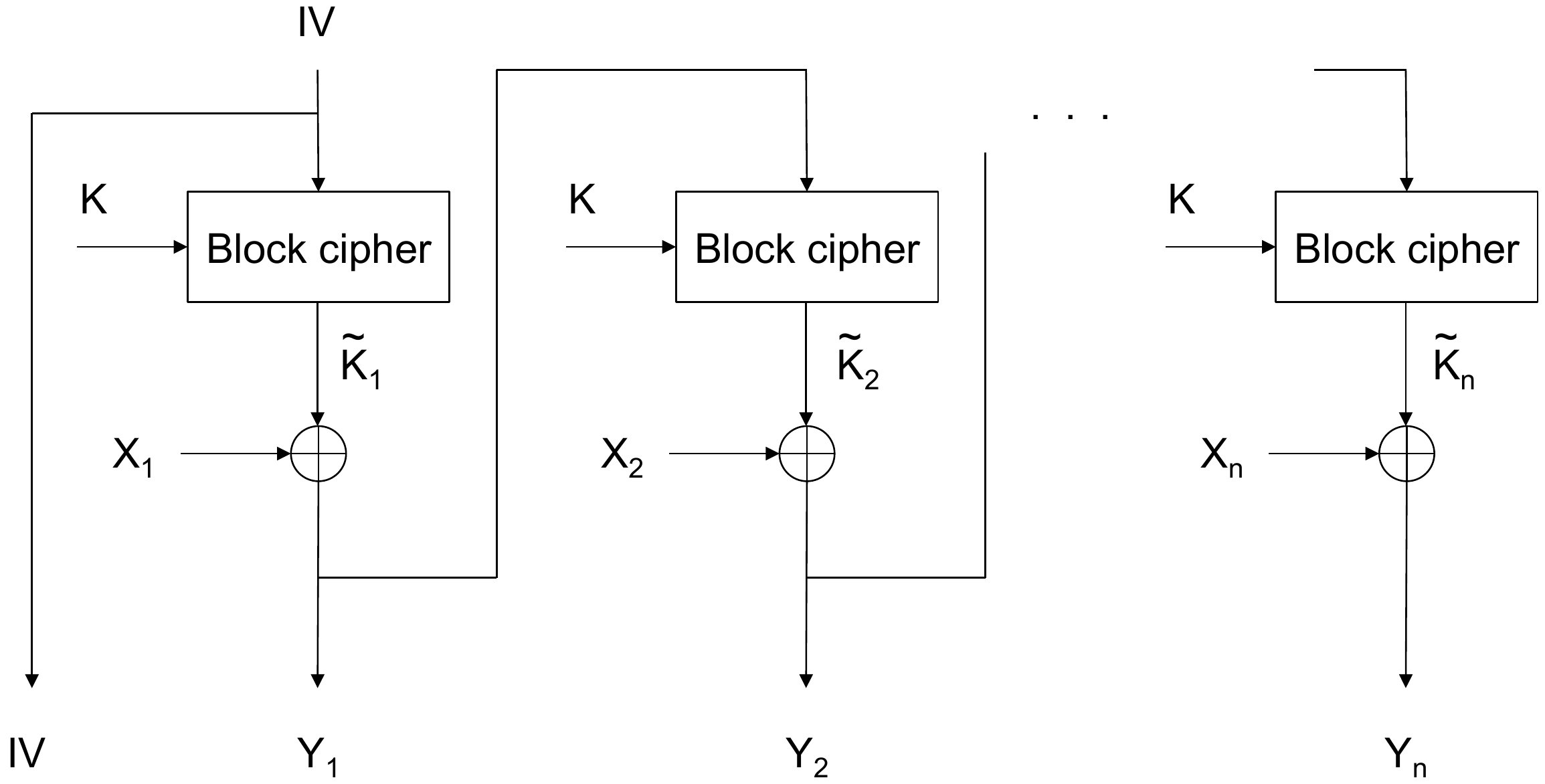}
\caption{Cipher feedback (CFB).}
\label{fig:cfb}
\end{figure}
% }%ignorefigure

However, in CFB mode, the pseudorandom blocks $\mathbf{\tilde{K}}^n$ are not independent of $\mathbf{X}^n$. Let the block
cipher using a secret key $K$ be characterized by the bijective mapping $B_K: \mathcal{X}^m \rightarrow \mathcal{X}^m$. Each block $\tilde{K}_i$ 
depends on the preceding plaintext block $X_{i-1}$ as follows:

\begin{equation}
\tilde{K}_i = B_K(X_{i}) = B_K(X_{i-1} \xor \tilde{K}_{i-1}),
\end{equation}
where $X_0$ is defined to be the IV.

% *********************************************************************************
%  Note: 
%  Can describe in more detail how decoding is performed !!!!
% *********************************************************************************

Due to the dependence between $\mathbf{\tilde{K}}^n$ and $\mathbf{X}^n$, the proposed compression algorithm for the CFB mode operates
on individual ciphertext blocks $Y_i$, rather then on $\mathbf{Y}^n$ all at once like in the OFB mode. Without this
distinction, the joint decompression and decryption as proposed in~\cite{Johnson04OnCompressing} would not be possible. 
Let $(C_\text{CFB},D_\text{CFB})$ denote a  Slepian-Wolf code
with encoding rate $R$ and block length $m$. The Slepian-Wolf encoding function
is defined as $C_\text{CFB}: \mathcal{X}^m \rightarrow \{1,\ldots,2^{mR}\}$, and the Slepian-Wolf decoding function as
$D_\text{CFB}:  \{1,\ldots,2^{mR}\} \times \mathcal{X}^m \rightarrow \mathcal{X}^m$. Compression is achieved by applying
the Slepian-Wolf encoding function $C_\text{CFB}$ to each of the ciphertext blocks in $\mathbf{Y}^n$ individually. The compressed
representation of $\ENC_K(\mathbf{X}^n)$ is then $(\text{IV}, C_\text{CFB}(Y_1),\ldots, C_\text{CFB}(Y_n))$. Note again, that $C_\text{CFB}$ does not require
knowledge of $K$.

Joint decompression and decryption, depicted on Figure~\ref{fig:decompressor_CFB} is performed sequentially from left to right, since decryption of the $i$th ciphertext block requires knowledge of the $(i-1)$th plaintext block. 
% \ignorefigure{
\begin{figure}[htb!]
\centering
\includegraphics[width=0.65\textwidth]{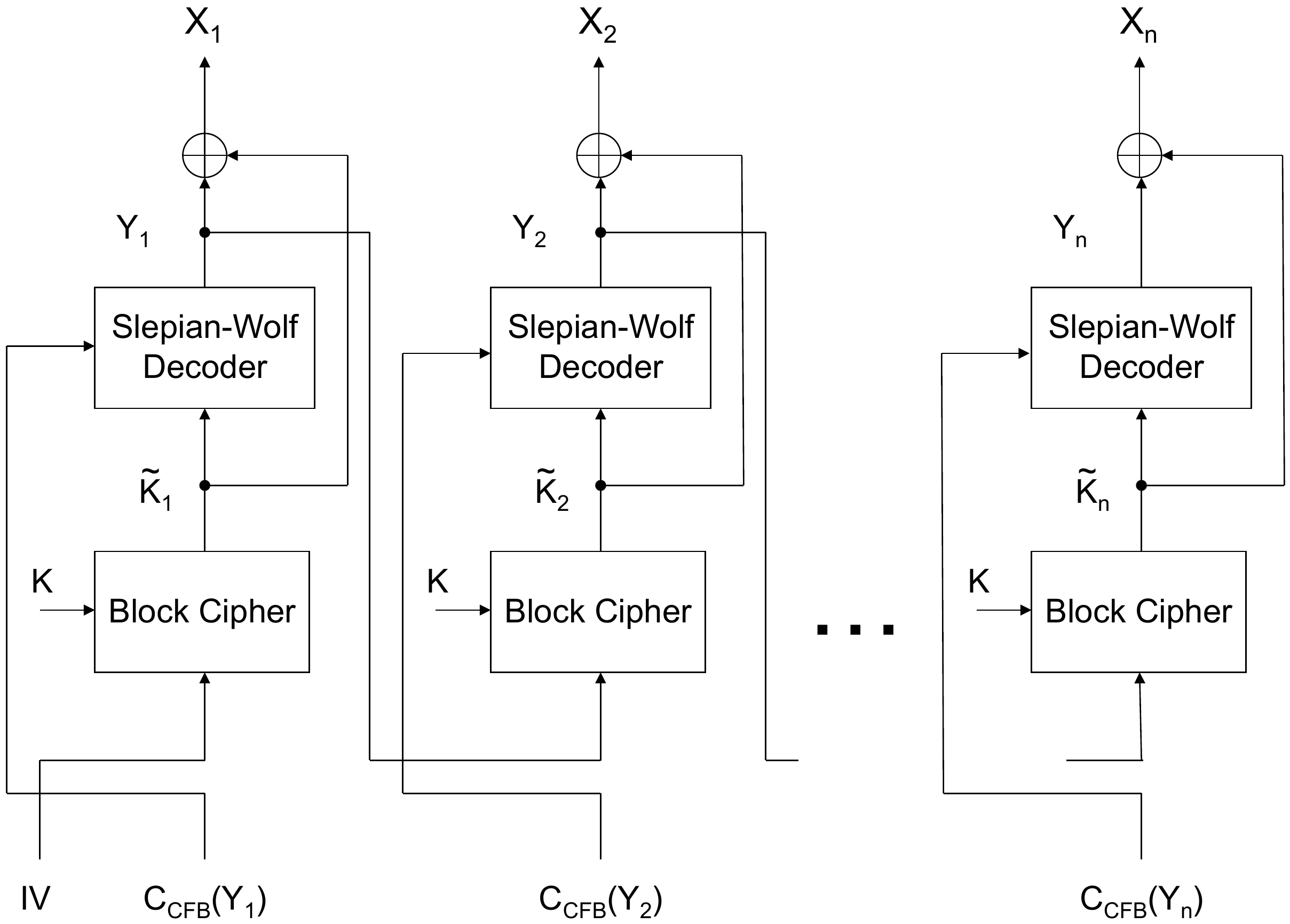}
\caption{Joint decryption and decoding in the CFB mode at the receiver is performed serially from left to right.}
\label{fig:decompressor_CFB}
\end{figure}
% }%ignorefigure
Initially, IV is mapped to $\tilde{K}_1$ by the block cipher. The Slepian-Wolf decoder function $D_\text{CFB}$ is applied 
to $(C_\text{CFB}(Y_1), \tilde{K}_1)$ to obtain $Y_1$. The first plaintext can now be obtained: $X_1 = \tilde{K}_1 \oplus Y_1$. 
Subsequently, $Y_1$ is mapped to $\tilde{K}_2$. The same process repeats to obtain $X_2$, and later all remaining plaintext blocks $\{X_i\}^n_{i=3}$.

Notice that in contrast to the OFB mode, this compression scheme is generally optimal only if $m \rightarrow \infty$. For finite $m$, the compression
is generally suboptimal, even if $n\rightarrow \infty$.

\section{Some Definitions from Cryptography}
\label{sec:cryptdefs}

For completeness, we recall some standard definitions from cryptography that are used in this paper.
See~\cite{Katz07Introduction} for details. These definitions
assume a fixed computational model over which time complexity is defined.

%Maybe you can change E_K notation to B_K and use E/D instead of Enc/Dec. Just a thought 
%toward a simpler notation.

\begin{definition}[Block Cipher]
A block cipher $B$ with block size $m$ is a keyed family $\set{B_K}_{K\in \K}$
where for each $K$, $B_K$ is a permutation over $m$ bits%
\footnote{ By a {\em permutation over $m$ bits} we mean a deterministic
bijective function over $\zo^m$.}
and $\K$ is the set of all possible keys.
\end{definition}

The security of a block cipher is defined via the notion of {\em indistinguishability}.
Ideally, we would like the behavior of a block cipher to be indistinguishable
by computational means from that of a purely random permutation over $m$ bits. 
However, since a block cipher is a much smaller family of permutations than the
family of {\em all} permutations, the above is not fully achievable. Yet, if we
restrict our attention to ``computationally feasible distinguishers" then
we can obtain a meaningful notion of security applicable to actual
block ciphers such as AES.

Hence, the main ingredient in such definition is that of a {\em distinguisher}.
A distinguisher $\Dist$ is defined as a randomized algorithm with oracle access to
two $m$-bit permutations $\ENC$ and $\DEC$ where $\DEC=\ENC^{-1}$. 
$\Dist$ can 
make arbitrary queries to the oracles and eventually outputs a single bit 0 or 1.
We consider the runs of $\Dist$ in two cases: When the oracles are instantiated
with a truly random permutation and when instantiated with a block cipher (i.e.,
with the functions $B_K$ and $B_K^{-1}$ where $K$ is chosen with uniform
probability from the set $\K$).
Let  $\Preal$ be the probability that $\Dist$ outputs a 1 when the oracle was 
instantiated with $B$, where $\Preal$ is computed over all random coins of $\Dist$ and all choices of the key for $B$. Further, let $\Prp$ 
be the probability that $\Dist$ outputs a 1 when the oracle was instantiated with a random
permutation, where $\Prp$ is computed over all random coins of $\Dist$ and all permutations $\ENC$.
Intuitively, we can think of $\Dist$ as trying to decide if the oracles are
instantiated with a random permutation or with a block cipher;
hence, a distinguisher is considered successful if the difference  $|\Preal-\Prp|$, which is referred to as
the {\sf advantage}, is non-negligible. Formally, this leads to the following
definition.%
\footnote{This definition corresponds to the notion of strong pseudorandom permutation \cite{Bellare97Concrete}.} 

\begin{definition}[Block Cipher Security]\label{def-blockcipher}
A block cipher B is called {\sf $(T,\eps)$-secure} if no distinguisher $\Dist$ that runs in
time T has advantage larger than $\eps$.
%running time considers each oracle call as a single operation
\end{definition}

This definition
tries to capture the idea that, for secure block ciphers, even distinguishers that have the ability to run
for extremely large time $T$, say $T=2^{80}$, gain only negligible distinguishing
advantage, say $\eps=2^{-40}$. In other words, for any practical purpose the
quality of the block cipher is as good as if it was instantiated by a ``perfect
cipher" (purely random permutation).

A secure block cipher by itself does not constitute a secure private-key encryption scheme due to its deterministic
nature. Namely, two identical plaintexts are mapped into two identical ciphertexts, therefore valuable information about 
data patterns can be leaked to eavesdroppers. Rather, secure block ciphers are used as building blocks that can be used
to construct private-key encryption schemes that eliminate this vulnerability. Note that a secure private-key encryption scheme 
must be probabilistic or stateful. 

Toward a formal definition of a secure encryption scheme, consider the following experiment:

\begin{definition}[CPA Indistinguishability Experiment]\label{def-cpa-experiment}
Consider an adversary $\A$ and an encryption scheme $(Gen,\ENC,\DEC)$. The {\em chosen plaintext attack (CPA) indistinguishability experiment} ${\sf Expt_\A^{cpa}}$
is defined as follows:
\begin{enumerate}
\item a key $k$ is generated by running $Gen$;

\item the adversary $\A$ has oracle access to $\ENC_K(\cdot)$, and queries it with a pair of plaintexts $X_0,X_1$ of the same length;

\item the oracle randomly chooses a bit $b\leftarrow\{0,1\}$ and returns the ciphertext $q\leftarrow\ENC_K(X_b)$, called the challenge,  to $\A$;

\item the adversary $\A$ continues to have oracle access to $\ENC_K(\cdot)$ and is allowed to make arbitrary queries. Ultimately it makes a guess about the value of $b$ by outputting $b'$;

\item the output of the experiment, ${\sf Expt_\A^{cpa}}$, is defined to be $1$ if $b' = b$, and $0$ otherwise. If ${\sf Expt_\A^{cpa}}=1$, we say that $\A$ succeeded.
\end{enumerate}
\end{definition}

\begin{definition}[Private-Key Encryption Security]\label{def-cpa}
A private-key encryption scheme $(Gen,\ENC,\DEC)$ is called {\sf $(T,\eps)$-indistinguishable under chosen plaintext attacks} (or CPA-secure) if for every adversary $\A$ that runs in time $T$,
$$
\Pr\Big[{\sf Expt_\A^{cpa}} = 1\Big] < \frac{1}{2} + \eps,
$$
where the probability is taken over all random coins used by $\A$, as well as all random coins used in the experiment.
\end{definition}

%Maybe some additional discussion -- this is a weaker condition than the one with block ciphers. Secure block ciphers are often building blocks of secure 
%encryption schemes. Randomization is achieved by the IV. Notice that if the encryption scheme must map two X0 to two different ciphertexts, otherwise
%the adversary could simply query X0 again. Since it is probabilistic, querying X0 or X1 again will not help...

%%%%%%%%%%

\section{Proof of Theorem \protect\ref{thm:lowerbound}}
\label{app:sec-proof}

%The core technical part of our lower bound argument is to show how to build a
%distinguisher for a block cipher out of a compression scheme that behaves
%noticeably better than just following one of the two exhaustive strategies
%above.

The proof of Theorem~\ref{thm:lowerbound} is by contradiction.
We assume a generic PEC scheme $(C,D)$ that 
departs in some noticeable way (later made more precise by means of the parameter $\eps$)
 from the exhaustive strategies when it is applied to a block cipher $B$. Subsequently, we show
how to use such a PEC scheme to build a distinguisher $\Dist$ that distinguishes
with advantage strictly larger than $\eps$ between the block cipher $B$ and a random
permutation, in contradiction to the security of $B$. 

%This appendix shows the following: if there exists a generic PEC scheme $(C,D)$
%for a specific block cipher $B$, whose coding 
%performance is noticeably better than the exhaustive strategies, we can build a
%distinguisher that breaks the security of $B$. We assume the existence of such
%scheme $(C,D)$ and reach a contradiction with the assumed security of $B$ (see
%Lemma~\ref{lemma-lowerbound} below).

%We start by stating the following simplifying assumptions on the behavior of
%the $(C,D)$ scheme (some of these assumptions do not limit the generality of
%our result; others will be discussed after the proof).

For simplicity, and without loss of generality, we assume that $D$ does not make redundant
queries to $\ENC$ or $\DEC$.
Namely, no query $X$ to $\ENC$ or query $Y$ to $\DEC$ is repeated in a run. If
$X$ was output by $\DEC$ (resp., $Y$ output by $\ENC$) it is not entered into $\ENC$
(resp., into $\DEC$).
In addition, if the decoded output from $D$ is $X$, we assume
that $X$ was either input to $\ENC$ or output by $\DEC$. If
 $D$ does not follow these rules, it can be modified to do so.

For $X\inr \P$ 
\footnote{We use $X\inr \P$ to denote that $X$ is chosen at random according to the probability distribution $\P$.} 
and $Y=\ENC(X)$, assume that $C(Y)$ is passed as input to $D$. We  say that $D$ decodes correctly
if it queries either $X$ from $\ENC$ or $Y$ from $\DEC$ during its run. 
In other words, $D$ is not required to have the ability
to identify the correct plaintext, which simplifies our presentation without 
weakening our results. On the contrary, it shows that even if such a relaxed
decoding requirement is acceptable, the lower bound we prove still holds.

Finally, note that the formulation of the theorem assumes that the plaintext
distribution $\P$ is efficiently samplable. This assumption is used in an
essential way in our proof, though the efficiency requirement from the $\P$
sampler is very weak.
In addition, we assume that, for a given compressed ciphertext $C(Y)$, one can
sample uniformly from the set $\Y_{C(Y)}=C^{-1}(C(Y))$, which is  the set of all ciphertexts
mapped by $C$ to $C(Y)$. Additional discussion related to this assumption is given after the proof.

%We assume that $\X$ is efficiently samplable, namely, there is an efficient
%probabilistic algorithm whose output distribution is $\X$.
%Say we need it for $\Dist$. The process for sampling could be given to $\Dist$ as a
%subroutine or oracle.  We see the ``real world" as a set of efficient processes
%so anything that is ``real" (supposedly $\X$, otherwise why we want to compress it)
%is efficienlty generateable or samplable.

The proof of Theorem \ref{thm:lowerbound} uses the following lemma.

\begin{lemma}
\label{lemma-lowerbound}
Let $T$ be a time-bound parameter, let $B$ be a $(T,\eps)$-secure block cipher and let $\E=(Gen,\ENC,\DEC)$ be an encryption scheme, where
$Gen$ chooses $K\inr \K$ and $\ENC=B_K$, $\DEC=B_K^{-1}$. Let $\P$ be a plaintext distribution samplable in time $T/4$ and let $(C,D)$
be a generic $(\E, \P,\delta)$-PEC scheme. Then either $(C,D)$ runs in time that
exceeds $T/4$ (and hence is infeasible\footnote{
For secure block ciphers, a distinguisher should not be able to attain more then
a negligible advantage even if it runs for extremely large time $T$, say
$T=2^{80}$ (see Appendix~\ref{sec:cryptdefs}). Thus, a PEC scheme that runs in
time $T/4$ would be considered infeasible.})
or the following holds:

%\samepage{
%Let $T$ be a time-bound parameter and
%let $(C,D)$ be a generic PEC-scheme for block ciphers
%over the plaintext distribution $\P$ which is samplable in time $T/4$.
%Let $B$ be a $(T,\eps)$-secure block cipher and assume that when
%instantiated with $B$ (i.e., $\ENC=B_K$ and $\DEC=B_K^{-1}$, for $K\inr \K$)
%and plaintext distribution $\P$, $D$ correctly decodes\footnote{We will later relax this to allow a small probability of failure by $D$.} ciphertexts
%compressed by $C$.
%}Then, either $(C,D)$ runs in time that exceeds $T/4$ and hence is
%infeasible\footnote{
%As we remarked at the end of Appendix~\ref{sec:cryptdefs}, one considers the
%security of a block cipher to resist distinguishers that run for extremely large
%time $T$, say $T=2^{80}$ (and have only negligible distinguishing advantage).
%Thus, a PEC scheme $(C,D)$ that runs time $T/4$ would still be considered infeasible.
%}
%or the following holds:

(i) Let $X \inr \P$ and $Y=\ENC(X)$. % and $c=C(Y)$.
Consider a run of $D$ on input $C(Y)$ in which $D$ queries $\ENC(X)$, and let
$X'$ be a random element drawn from $\P$ independently of $X$.
Then, the probability that $D$ queries $\ENC(X')$ before $\ENC(X)$
is at least $1/2-\eps/(1-\delta)$.

%[TO BE USED IN THE PROOF: we will show that if the above does not hold then when
%running B as $\ENC$, $\Dist$ will see $x$ called before $x'$ with probability larger than
%$1/2+\eps$, while when running RP it will see it with probability exactly $1/2$ does
%having an advantage larger than $\eps$.]

(ii) Let $X \inr \P$ and  $Y=\ENC(X)$. % and $C(Y)$.
Consider a run of $D$ on input $C(Y)$ in which $D$ queries $\DEC(Y)$, and let
$Y'$ be an element drawn uniformly from $\Y_{C(Y)}=C^{-1}(C(Y))$.
Then, the probability that $D$ queries $\DEC(Y')$ before $\DEC(Y)$
is at least $1/2-\eps/(1-\delta)$.

\end{lemma}

\smallskip

We first show how the Lemma~\ref{lemma-lowerbound} suffices to prove Theorem~\ref{thm:lowerbound}.

%\medskip

\begin{proof}[Proof of Theorem~\ref{thm:lowerbound}]
%{\it Proof of Theorem \protect\ref{thm:lowerbound}:}
For case (i), the
probability that $X$ is queried first is within $\eps/(1-\delta)$ of the probability that $X'$ is queried first,
%strictly speaking the probability to query $x$ first is at most $1/2+\eps$ so it could
%even be much less than 1/2 but this just means that $D$ is even worse than the
%exhaustive search strategy
where the latter is the probability of querying a plaintext that bears no
information (the run of $D$ is independent of $X'$). Assuming that $\delta<1/2$, we get $\eps/(1-\delta) < 2\eps$ and since $\eps$ is negligible (say $2^{-40}$) so is $2\eps$. 
This implies a plaintext-exhaustive strategy by $D$.
%To formalize this I should say that this gives a lower bound on $D$ related
%to the GUESSING ENTROPY of $\X$ (which can be achieved only by enumeration).
%(see Cachin97Entropy).
Similarly, for case (ii), the probability that the value $Y$ is computed first is within $\eps/(1-\delta)$ of the 
probability that an independent $Y' \inr \Y_{C(Y)}$ is computed first. This implies a ciphertext-exhaustive strategy by $D$.
\end{proof}
\ignore{ I MAY WANT TO SAY THE FOLLOWING:
Therefore the best strategy is to query the ciphertexts that map to $C(Y)$ under
some enumeration that maximizes the probability to guess $Y$ correctly.
However, from the proof of the lemma and the indistinguishability of $B$ from
a random permutation it follows that there is no efficient enumeration
that will find the correct $Y$ with probability significantly better than
$1/|\Cinv|$.
%This is the case even if the enumerating algorithm has full knowledge of the
%distribution $\X$, including the probability assigned to each $x$ and a
%ordering of $x$ in decreasing order of probability.
}

\medskip 

\begin{proof}[Proof of Lemma~\ref{lemma-lowerbound}]
%{\it Proof of Lemma \protect\ref{lemma-lowerbound}:}
First, consider the error-free case, i.e. $\delta = 0$. We show that if $(C,D)$ runs in time less
than $T/4$ and conditions (i), (ii) do not hold, we can build a
distinguisher (see Appendix~\ref{sec:cryptdefs}) that runs in time at most $T$ and distinguishes between $B$ and
a random permutation with an advantage larger than $\eps$, in contradiction to
the security of $B$. A distinguisher  interacts with oracles $\ENC$ and $\DEC$ and its
goal is to identify whether the
oracles are instantiated with the real block cipher $B$ or with a random
permutation. 
For clarity, we represent the output 0 from $\Dist$ by the symbol $\RP$
($\Dist$ decided that the oracles are
instantiated with a random permutation) and the output 1 by $\REAL$ ($\Dist$ decided that the oracles are instantiated by the block
cipher $B$).
%The output 0 from a distinguisher is represented by the symbol $\RP$, which
%indicates its belief that the oracles are instantiated with a random
%permutation, and the output 1 by $\REAL$, which indicates that it believes the
%oracles are instantiated with the  block cipher $B$. 

%The idea is that $\Dist$ monitors the queries to $\ENC$ and $\DEC$ as
%requested by $(C,D)$ and as soon as it observes a departure from (i) or (ii) it
%decides on its output.

We build a distinguisher $\Dist$ that uses the scheme $(C,D)$ and 
responds to the encryption and decryption queries made by $D$ with its $\ENC/\DEC$ oracles. When  $\ENC$ is a random permutation,
$(C,D)$ may run much longer than when $\ENC$ is $B$.
Thus, to bound the time complexity of $\Dist$, we set a time limit $T'=T/4$,
such that if the total time of $(C,D)$ exceeds $T'$,
$\Dist$ stops as well. As we show below, the parameter $T'$ is chosen as
$T/4$ to ensure that the total running time of $\Dist$ is no more than $T$.
%we need to make an assumption on the cost of sampling $x',y'$ - remark after
%the lemma.

Initially, $\Dist$ chooses $X \inr \P$ and receives  the
value $Y=\ENC(X)$ from its $\ENC$ oracle.
% here we assume samplability of $\X$.
It computes $C(Y)$ and passes it as input to  $D$.
In addition, $\Dist$ chooses an independent $X'\inr \P$ and independent
$Y'\inr \Y_{C(Y)}$. It is assumed that $Y'$ can be sampled within time
$T/4$.
%This assumes such samplability is efficient, see after the lemma for an
%alternative.
Subsequently, $\Dist$ monitors the queries to $\ENC/\DEC$ as requested by $D$ and
reacts to the following events:

\begin{enumerate}
\item if $X$ is queried from $\ENC$, stop and output $\REAL$;
\item if $X'$ is queried from $\ENC$, stop and output $\RP$;
\item if $Y$ is queried from $\DEC$, stop and output $\REAL$;
\item if $Y'$ is queried from $\DEC$, stop and output $\RP$.

%In the case that I can't sample $y'$ but know the size $S$ of $\Cinv$, I will
%need to replace (3), (4) with:\\
%(3') If $y$ is queried from $\DEC$ before $S/2$ $\DEC$ queries, stop and
%output($\REAL$)\\
%(4') If $y$ is queried from $\DEC$ after $S/2$ $\DEC$ queries, stop and
%output($\RP$)

\item if the run of $(C,D)$ exceeds $T'$, stop and output $\RP$
\end{enumerate}

It is possible that multiple such events take place in one run of $D$, for instance both $X$
and $X'$ may be queried from $\ENC$. In such case $\Dist$ stops as soon as it identifies 
first such event.

Let $\Preal$ and $\Prp$ be defined as in Appendix~\ref{sec:cryptdefs}. We
evaluate the advantage of $\Dist$, namely, the difference $|\Preal - \Prp|$.
First, consider a run of $\Dist$ when $\ENC/\DEC$ are instantiated by a random permutation.
The behavior 
of $(C,D)$ depends on $Y$ which is chosen at random and independently of $X$ and
$X'$, therefore the run is independent of both $X$ and $X'$. In effect, the probability that $X$ is queried before $X'$ 
is exactly 1/2. Similarly, if $Y$ is queried from $\DEC$,
 the behavior of $D$ depends
on $C(Y)$, while both $Y$ and $Y'$ have the same probability to be the chosen
as the preimage of $C(Y)$. Therefore, the probability that $Y$ precedes $Y'$ is exactly 1/2.
%i.e., $y$ is uniformly distributed over $\Cinv$)
%[OR: change ``before $y'$" with ``before $S$ $\DEC$ queries"]
It follows that $\Dist$ outputs $\REAL$ with probability at most 1/2 (exactly 1/2
for the $X$ and $Y$ cases and with probability 0 if $\Dist$ exceeds time $T'$),
i.e., $\Prp\leq 1/2$.

Now, consider a run of $\Dist$ when $\ENC/\DEC$ are instantiated by a block cipher $B$ and its inverse, respectively. 
Assume, for contradiction, that $(C,D)$ stops before time $T'$ and either
(i) the probability that $X'$ is queried before $X$ is strictly less than $1/2-\eps$
or
(ii) the probability that $Y'$ is queried before $Y$
%[$y$ queried before $S$ $\DEC$ queries]
is strictly less than $1/2-\eps$.  
%It follows that the probability that $\Dist$ outputs $\RP$ is strictly smaller
%than $1/2-\eps$ and the probability of $\REAL$ is strictly larger than
%$1/2+\eps$.
It follows that the probability that $\Dist$ outputs $\RP$ is strictly smaller
than $1/2-\eps$, therefore $\Preal>1/2+\eps$.

Thus, $\Dist$ distinguishes with advantage $|\Preal-\Prp|$, which is strictly larger than $\eps$.
The running time of $\Dist$ is upper-bounded by $T$: it includes three samplings (of $X,X'$ and $Y'$),
each assumed to take at most $T/4$ time, and the work of $(C,D)$ which $\Dist$
runs for total time $T/4$ at most.
In all, we have built a distinguisher against $B$ that runs time $T$ and has
advantage larger than $\eps$ in contradiction to the security of the block cipher $B$.

Assume now that $\delta > 0$. When a positive probability of error for $D$ is allowed,
it can occur that $D$ 
stops before time $T'$ and before $\Dist$ sees $X,X',Y$ or $Y'$.
To deal with this situation, we add a clause to the specification of $\Dist$
saying that if $(C,D)$ stops before time $T'$ and before
seeing any of the values $X,X',Y,Y'$, then $\Dist$ chooses a random bit $b$ and
outputs $\REAL$ if $b=1$ and $\RP$ if $b=0$.
We slightly increased the probability that $X'$ is queried before
$X$ (or $Y'$ before $Y$) in the block cipher case, however it is still
negligibly far from $1/2$. 
The proof is now a straightforward extension of the case when $\delta = 0$.
\end{proof}

%DO NOT ERASE
%FOR ME: Proof of the $1/2-\frac{\eps}{1-\delta}$ bound.
%As in the lemma we only need to consider the case where $(C,D)$ runs under the
%time bound $T'$.  We have that
%$$Prob(\REAL)=Prob(\REAL~and~SUCCESS) + Prob (\REAL~and~FAIL)=$$
%$$Prob(\REAL :~SUCCESS) Prob(SUCCESS) + Prob(\REAL:~FAIL) Prob(FAIL)=$$
%$$\pi (1-\delta) + Prob(b=1) \delta = \pi (1-\delta) + \delta/2$$
%where $\pi=Prob(x~seen~before~x'~or~y~seen~before~y')$
%Now, in the random permutation case, $\pi=1/2$ and thus $Prob(\REAL)=1/2$
%while for the block cipher $B$ case we will be assuming (for contradiction)
%$\pi>1/2-\eps'$ where $\eps'=\frac{\eps}{1-\delta}$ and then
%$Prob(\REAL) > (1/2-\eps')(1-\delta)+1/2 \delta = 1/2-\eps'+(\eps' \delta) =
%= 1/2-\eps'(1-\delta) = 1/2-\eps$.
%Thus showing that $\Dist$ distinguishes between a random permutation and block
%cipher $B$ with advantage > $\eps$.

\medskip
\noindent{\bf Remark.}
The proof of Lemma~\ref{lemma-lowerbound} assumes that the set $\Y_{C(Y)}$ is samplable in time $T/4$. 
While this assumption is likely to hold, we note that it is
enough to know the size of $\Y_{C(Y)}$. 
In such case, rather than sampling $\Y_{C(Y)}$, the events 3) and 4)
in the proof are replaced with the following one:
 if the number of queries to $\DEC$ performed by $D$
exceeds $|\Y_{C(Y)}|/2$ before $Y$ is queried, stop and output $\RP$.
%IF $\Dist$ knows $S$ exactly we can change the test ``$y'$ before $y$"
%with ``$y$ is queried from $\DEC$ before $S/2$ $\DEC$ queries" (see 3',4' above).
%

We now sketch the proof of the theorem when none of the above conditions holds,
namely when $\Y_{C(Y)}$ is of unknown size and not samplable in time $T/4$.
Assume that $(C,D)$ performs noticeably
better than the exhaustive strategies when the $\ENC/\DEC$ oracles are instantiated 
with the block cipher $B$.
Then it must hold that $(C,D)$ runs noticeably faster when the $\ENC$ and $\DEC$ oracles are 
instantiated with the block cipher $B$ than with a random permutation, for 
the exhaustive
strategies are optimal for a random permutation (as shown above).
Using this (assumed) discrepancy between the runs of $(C,D)$ over $B$ and the
runs of $(C,D)$ over a random permutation, we can build a distinguisher against
$B$ in contradiction to the security of $B$. In the following, we formalize this
discrepancy and outline the construction of the distinguisher, where some
straightforward details are omitted.

For any plaintext $X$ let $T_B(X)$ denote the runtime of $(C,D)$ on input $X$
%chosen according to $\P$ 
when $\ENC/\DEC$ oracles are instantiated with the block cipher $B$. Further, let
$T_R(X)$ denote the runtime of $(C,D)$ on input $X$ when $\ENC/\DEC$
oracles are instantiated with a random permutation.
%NOTE: T_B(X) and T_R(X) are random variables 
We assume that there is a known time bound $T_0$ such that $T_B(X)<T_0$ for all
$X$ (we relax this assumption below) 
and there exists a non-negligible $\eps$ such that 
%for $X \inr \P$,
$ \Pr[T_B(X) < T_R(X)] \geq 1/2+\eps$.
The probability is over all choices of $X$ and all random coins of $(C,D)$. Further, it 
is over all key choices for $B$ for $T_B$ and over all random permutations for $T_R$.
%NOTE FOR US: In a more formal argument, we start from the assumption that $B$
%is $(T,\eps)$ and show that either $T_B(X)$ exceeds $T_0=T/2$ or the difference
%between  $T_B$ and $T_R$ is no more than $\eps$.
We build a distinguisher $\Dist$ as follows:

\begin{enumerate}
\item
Choose $X \inr \P$ and run $(C,D)$ using the input oracles.
If time $T_0$ is exceeded, stop and output $\RP$,
otherwise proceed to Step 2.
\item
Let $T_1$ denote the running time of $(C,D)$ in step (1).
Run $(C,D)$ again on $X$ (same $X$ as in step 1)) but this time ignore the given
$\ENC/\DEC$ oracles. Instead, answer queries from $(C,D)$ with a random
permutation. If time $T_1$ is exceeded then output $\REAL$, otherwise output $\RP$.
\end{enumerate}

\noindent
We have the following: 
\begin{itemize}
\item If the $\ENC/\DEC$ oracles are instantiated with $B$, step (1) always
completes and in step (2) $\REAL$ is output with the probability that equals
$\Pr[T_R(X)>T_B(X)]$, which is at least $1/2+\eps$. 
%Thus: Prob(CIPHER:cipher)> 1/2+\eps

\item If the $\ENC/\DEC$ oracles are instantiated with a random permutation,
$\REAL$ is output only if $T_R(X)\leq T_0$ and the runtime on $X$ in step (2)
exceeds the runtime on $X$ in step (1). The probability of the latter is at most
1/2 since both runs are over a random permutation.
% it is not necessarily exactly 1/2 since it may be that the time in both steps
% is equal
%Thus: $\Pr(CIPHER:random) \leq 1/2$.
\end{itemize}
Thus, $\Dist$ will output $\REAL$ with probability at least $1/2+\eps$
when the oracles are instantiated with $B$,
while it will output $\REAL$ with probability at most $1/2$
when the oracles are instantiated with a random permutation.
It follows that $\Dist$ is a $(T,\eps)$-distinguisher for $B$ where $T=2T_0$
since in each of the steps (1) and (2) $\Dist$ runs time at most $T_0$.

Note that  the requirement that $T_B(X)<T_0$ for all $X$ can be relaxed such that
the joint probability of $T_B(X)<T_0$ and $T_B(X)<T_R(X)$ is at least $1/2+\eps$.

\ignore{ OLD TEXT:
$\Dist$ needs to use a statistical test comparing the behavior of
its input oracles against a run where the oracles are random (note that
$\Dist$ can easily simulate the latter).
THE FOLLOWING WAS COMMENTED OUT IN THE OLD TEXT TOO
%will need to do sampling of input oracles $\ENC$,$\DEC$ and compare running
%time with a simulated RP run, and if $\ENC$/$\DEC$ ``consistently wins" then
%decide $\ENC$/$\DEC$ are real.  I have not carried the analysis in this case.
}

\bibliographystyle{IEEEtran}
\bibliography{mybiblio}
\end{document}